\documentclass[%
 reprint,
groupedaddress,
 amsmath,amssymb,
 aps,
pra,
]{revtex4-2}

\usepackage{graphicx}
\usepackage{dcolumn}
\usepackage{bm}
\usepackage{comment}
\usepackage{amsthm}
\usepackage{color}

\usepackage{braket}
\usepackage{subcaption}

\newtheorem{theorem}{\bf Theorem}[section] \newtheorem{definition}{\bf Definition}[section]
 \newtheorem{remark}{\bf Remark}[section]
  \newtheorem{proposition}{\bf Proposition}[section]

\newtheorem{problem}{\bf Problem}[section]

\usepackage{bbm}

\newcommand\rmd{\mathrm{d}}

\newcommand\rmx{\mathrm{x}}
\newcommand\rmy{\mathrm{y}}
\newcommand\rmz{\mathrm{z}}

\newcommand\rmZ{\mathrm{Z}}

\newcommand\bbc{\mathbb{C}}

\newcommand\bbr{\mathbb{R}}

\newcommand\bbu{\mathbb{U}}

\newcommand\calC{\mathcal{C}}

\newcommand\calE{\mathcal{E}}
\newcommand\calF{\mathcal{F}}

\newcommand\calM{\mathcal{M}}

\begin{document}

\preprint{APS/123-QED}

\title{Robustness of quantum algorithms against coherent control errors}

\author{Julian Berberich}
 \affiliation{University of Stuttgart, Institute for Systems Theory and Automatic Control, 70569 Stuttgart, Germany}
 \email{julian.berberich@ist.uni-stuttgart.de}
 \author{Daniel Fink}%
 \affiliation{%
 University of Stuttgart, Institute for Computational Physics, 70569 Stuttgart, Germany
 }%
\author{Christian Holm}%
\affiliation{%
University of Stuttgart, Institute for Computational Physics, 70569 Stuttgart, Germany
}%

\date{\today}

\begin{abstract}
    Coherent control errors, for which ideal Hamiltonians are perturbed by unknown multiplicative noise terms, are a major obstacle for reliable quantum computing.
    In this paper, we present a framework for analyzing the robustness of quantum algorithms against coherent control errors using Lipschitz bounds.
    We derive worst-case fidelity bounds which show that the resilience against coherent control errors is mainly influenced by the norms of the Hamiltonians
    generating the individual gates.
    These bounds are explicitly computable even for large circuits, and they can be used to guarantee fault-tolerance via threshold theorems.
    Moreover, we apply our theoretical framework to derive a novel guideline for robust quantum algorithm design and transpilation,
    which amounts to reducing the norms of the Hamiltonians.
    Using the $3$-qubit Quantum Fourier Transform as an example application, we demonstrate that this
    guideline targets robustness more effectively than existing ones based on circuit depth or gate count.
    Furthermore, we apply our framework to study the effect of parameter regularization in variational quantum algorithms.
    The practicality of the theoretical results is demonstrated via implementations in simulation and on a quantum computer.
\end{abstract}

\maketitle


\section{Introduction}\label{sec:intro}
Quantum computing has emerged as a powerful tool to overcome limitations of classical computing and solve problems that were previously intractable.
Much research has been devoted to developing algorithms which yield provable speedups over their classical counterparts, including, e.g., integer factoring~\cite{shor1999polynomial} or search algorithms~\cite{grover1996fast}.
However, a successful practical implementation of these algorithms for relevant problem sizes often requires large circuits with many reliable qubits and gates, which are not available in the current noisy intermediate-scale quantum (NISQ) era~\cite{preskill2018quantum}.
In particular, current quantum circuits are affected by significant amounts of noise, which poses a key challenge for demonstrating any quantum advantage.

Noise occurring on quantum devices can be categorized into decoherent vs.\ coherent errors.
An error is coherent if it can be written as a unitary operator, and is decoherent otherwise.
In this paper, we deal with \emph{coherent control errors}, which are an important type of coherent errors.
Mathematically, a coherent control error can be described as a perturbation
\begin{align}\label{eq:intro_coherent_control_error}
    e^{-i(1+\varepsilon)H}=e^{-iH}e^{-i\varepsilon H}
\end{align}
of an ideal quantum gate $e^{-iH}$, where $H=H^\dagger$ is the Hamiltonian generating the gate.
Further, $\varepsilon\in\bbr$ is an unknown, possibly stochastic and/or time-varying noise term.
On the hardware level, such errors can be caused by imprecise classical control, e.g., due to miscalibration or imperfect actuation.
Coherent control errors have been recognized as a crucial error source on current quantum hardware~\cite{barnes2017quantum,trout2018simulating,arute2019quantum}.
As a result, different approaches have been proposed to study and compensate coherent control errors, e.g., 
using composite pulses~\cite{levitt1986composite,jones2003robust,brown2004arbitrarily,merrill2014progress}, dynamically error-corrected gates~\cite{khodjasteh2009dynamically},
quantum error correction~\cite{bravyi2018correcting,debroy2018stabilizer,edmunds2020dynamically,majumder2020real,liu2022exact}, randomized compiling~\cite{wallman2016noise}, gate set tomography~\cite{nielsen2021gate}, hidden inverses~\cite{zhang2022hidden}, or experimental calibration~\cite{maksymov2021optimal}.
    Although independent coherent control errors are particularly detrimental~\cite{barnes2017quantum}, the majority of the above approaches focuses on systematic errors, i.e., errors where the precise value of $\varepsilon$ in~\eqref{eq:intro_coherent_control_error} is constant in time and among the gates.

\begin{figure}
    \begin{center}
    \includegraphics[width=0.42\textwidth]{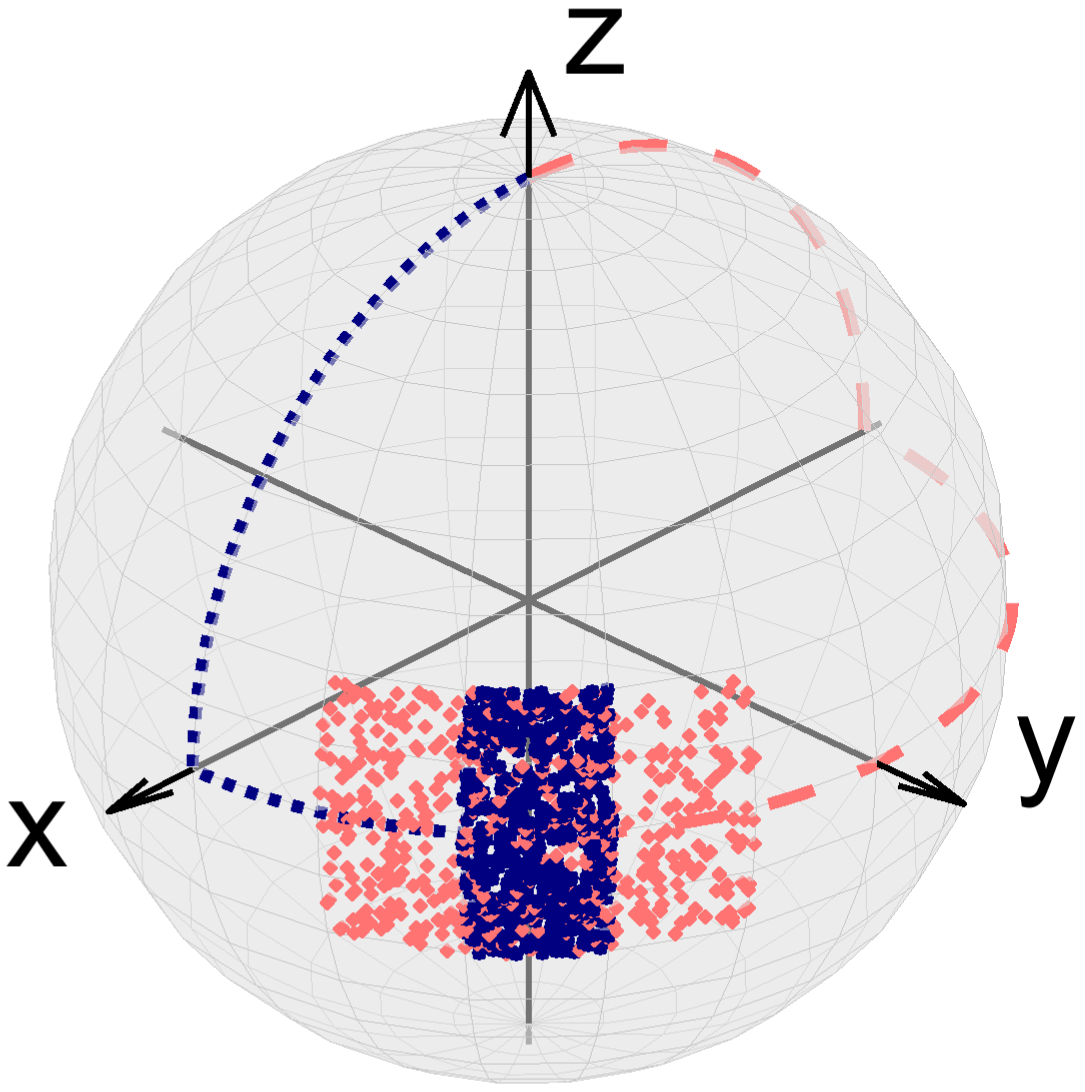}
    \end{center}
    \caption{
        The figure shows the evolution of the quantum state resulting from applying 
        $U_{\mathrm{ideal}}=R_{\rmz}(\frac{\pi}{4})R_{\rmy}(\frac{\pi}{2})$ (blue, dotted) and $U_{\mathrm{ideal}}'=R_{\rmz}(-\frac{3\pi}{4})R_{\rmy}(-\frac{\pi}{2})$ (red, dashed) to $\ket{0}$.
        Additionally, the final states for random realizations of coherent control errors are shown for both circuits (in blue and red, respectively).
        }
    \label{fig:bloch_H}
\end{figure}

In this paper, we develop a framework for analyzing the robustness of quantum algorithms against coherent control errors based on Lipschitz bounds.
We provide novel theoretical insights into the effect of such errors, which can be used to quantitatively assess and improve 
the inherent robustness of quantum algorithms.
Let us motivate our framework with an illustrative example depicted in Figure~\ref{fig:bloch_H}.
Suppose we want to apply the sequence of rotations 
\begin{align}
    U_{\mathrm{ideal}}=R_{\rmz}\Big(\frac{\pi}{4}\Big)R_{\rmy}\Big(\frac{\pi}{2}\Big)
\end{align}
to $\ket0$, where 
$R_{\rmy}(\theta_{\rmy})$ denotes the rotation around the $y$-axis with angle $\theta_{\rmy}$, and similarly for $R_{\rmz}(\theta_{\rmz})$.
Alternatively, the same final state can be reached (modulo a global phase) by applying
\begin{align}
    U_{\mathrm{ideal}}'=R_{\rmz}\Big(-\frac{3\pi}{4}\Big)R_{\rmy}\Big(-\frac{\pi}{2}\Big)
\end{align}
to $\ket0$.
Figure~\ref{fig:bloch_H} shows the evolution of both quantum circuits.
We see that both circuits produce the same final state $\ket{\hat{\psi}}$ if there is no noise.
From a purely computational or algebraic perspective, it is not immediately obvious which of the two is preferable, e.g.,
both use the same gate sets ($y$- and $z$-rotations) and both have the same gate count and circuit depth.

Let us now compare these circuits in terms of robustness.
More precisely, we assume that all involved gates are affected by independent coherent control errors, i.e., rather than $R_{\rmy}(\theta_{\rmy})$ and $R_{\rmz}(\theta_{\rmz})$, 
we can only implement
\begin{align}
    R_{\rmy}(\theta_{\rmy} (1+\varepsilon_\rmy))
    &=
    e^{-i (1+\varepsilon_\rmy) \frac{1}{2} \theta_\rmy Y},
    \\
    R_{\rmz}(\theta_{\rmz} (1+\varepsilon_\rmz))
    &=
    e^{-i (1+\varepsilon_\rmz) \frac{1}{2} \theta_\rmz Z},
\end{align}
where $\varepsilon=(\varepsilon_\rmy,\varepsilon_\rmz)\in\bbr^2$ contains unknown noise terms as in~\eqref{eq:intro_coherent_control_error}.
Figure~\ref{fig:bloch_H} shows the final states when applying each of the two unitaries
\begin{align}
    U_{\mathrm{noisy}}(\varepsilon)=R_{\rmz}\Big(\frac{\pi}{4}(1+\varepsilon_\rmz)\Big)R_{\rmy}\Big(\frac{\pi}{2}(1+\varepsilon_\rmy)\Big)
\end{align}
and 
\begin{align}
    U_{\mathrm{noisy}}'(\varepsilon')=R_{\rmz}\Big(-\frac{3\pi}{4}(1+\varepsilon_\rmz')\Big)R_{\rmy}\Big(-\frac{\pi}{2}(1+\varepsilon_\rmy')\Big)
\end{align}
to $\ket0$ for $500$ different values of $\varepsilon,\varepsilon'\in\bbr^2$, which are randomly sampled from a uniform distribution over $[-0.2,0.2]\times[-0.2,0.2]$.

While the circuits $U_{\mathrm{ideal}}$ and $U_{\mathrm{ideal}}'$ are identical in the absence of noise, the coherent control errors cause a substantial difference.
In particular, $U_{\mathrm{noisy}}'$ is affected by a significantly larger over- or under-rotation around the $z$-axis in comparison to $U_{\mathrm{noisy}}$.
We can also make a quantitative comparison by computing the fidelity of each perturbed final state $\ket{\tilde{\psi}}$ w.r.t.\ 
the ideal one $\ket{\hat{\psi}}$, i.e., $|\braket{\tilde{\psi}|\hat{\psi}}|$, and by determining the minimum fidelity over all $500$ noise realizations.
The resulting value for $U_{\mathrm{noisy}}$ is $0.985$, whereas $U_{\mathrm{noisy}}'$ yields a minimum fidelity of $0.965$.
This shows that the loss of fidelity of $U_{\mathrm{noisy}}'$ due to the worst-case realization of the coherent control error is more than twice 
as large as that of $U_{\mathrm{noisy}}$, i.e., the latter circuit
is significantly more robust against such errors.
The results presented in this paper allow to explain this observation and much more general scenarios by precisely quantifying the robustness
of quantum algorithms depending on the involved gates.

\subsection*{Contribution}
In this paper, we develop a framework for robustness analysis of quantum algorithms against coherent control errors.
More precisely, we use Lipschitz bounds to derive worst-case fidelity bounds against coherent control errors, which
depend only on the components of the given circuit and are explicitly computable.
In particular, we show that the resilience of a circuit against coherent control errors mainly depends on the norms 
of the Hamiltonians as well as on the coupling between sequentially applied gates.
The presented results are applicable under rather general conditions:
We allow for independent noise terms affecting the gates, for arbitrary unitaries defining the gates,
and we do not require any assumptions on the nature of the noise, e.g., being sufficiently small or drawn from a
specific probability distribution.

We apply our theoretical framework to the following problems.
First, we derive explicit worst-case bounds for coherent control errors, which can be used in quantum error correction (QEC)
threshold theorems.
Moreover, we propose a novel guideline for robust quantum algorithm design and transpilation, which amounts to reducing the
norms of the involved Hamiltonians.
The norms of the Hamiltonians can provide a more accurate robustness measure than existing ones based on circuit depth or gate count 
as commonly used in the literature on circuit optimization and transpilation~\cite{maslov2008quantum,arabzadeh2010rule,nam2018automated,amy2019controlled,lee2019hybrid,duncan2020graph,foesel2021quantum,nagarajan2021quantum}.
Thus, the proposed guideline leads to quantum circuits which are inherently more robust against coherent control errors and, thereby, 
more easily implementable in the near term.
We illustrate this principle by studying the robustness of the $3$-qubit Quantum Fourier Transform (QFT) when transpiled into
different elementary gate sets.
Finally, we apply our results to variational quantum algorithms (VQAs), where we show that 
parameter regularization improves the robustness against coherent control errors.
Our theoretical findings are confirmed by simulations and with an implementation on the \textit{ibm\_nairobi} quantum computer~\cite{ibm_hardware}.

\subsection*{Outline}
In Section~\ref{sec:fidelity}, we derive worst-case fidelity bounds for quantum circuits affected by coherent control errors based on Lipschitz bounds.
These theoretical results are connected to threshold theorems for fault-tolerant quantum computing (Section~\ref{sec:fault_tolerance}),
and they are used to derive a novel guideline for robust quantum algorithm design and transpilation (Section~\ref{sec:guidelines}).
In Section~\ref{sec:robustness_elementary}, we use our framework to analyze the robustness of different elementary gate
decompositions of the $3$-qubit QFT.
Section~\ref{sec:validation} contains a validation of our results on a quantum computer and
Section~\ref{sec:vqa} addresses parameter regularization in VQAs.
Finally, we conclude the paper in Section~\ref{sec:conclusion}.
The source code for all of the performed experiments is publicly accessible on GitHub~\cite{code_repository}.

\subsection*{Notation}
We write $\lVert x\rVert_p$ for the $p$-norm of a vector $x\in\bbc^N$ and $\lVert A\rVert_2$ for the induced $2$-norm (i.e., the maximum singular value) of a matrix $A\in\bbc^{N\times N}$.
Further, the maximum eigenvalue of a Hermitian matrix $A=A^\dagger$ is denoted by $\lambda_{\max}(A)$.
We denote the $N$-dimensional identity matrix by $I_N$ and the $N$-dimensional unitary group by $\bbu^N$.

\section{Worst-case bounds for coherent control errors}\label{sec:fidelity}

\begin{figure}
    \centering 
    \begin{subfigure}{0.75\columnwidth}
        \centering
        \includegraphics[width=0.85\textwidth]{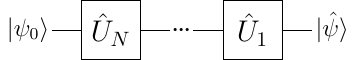}
        \caption{Ideal quantum circuit}
    \end{subfigure}
\vskip5pt
    \begin{subfigure}{1.045\columnwidth}
        \centering
        \includegraphics[width=0.85\textwidth]{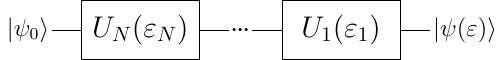}
        \caption{Noisy quantum circuit}
    \end{subfigure}
    \caption{Ideal and noisy quantum circuits.}	\label{fig:circuit_ideal_noisy}
\end{figure}

Consider the following \emph{ideal} quantum circuit
\begin{align}\label{eq:qc}
    \ket{\hat{\psi}}=\hat{U}_1\cdots \hat{U}_N\ket{\psi_0},
\end{align}
consisting of $N$ unitary operators $\hat{U}_i\in\bbu^{2^n}$ acting on the initial state $\ket{\psi_0}$, cf.\
Figure~\ref{fig:circuit_ideal_noisy} (a).
These unitaries can also be written as $\hat{U}_i=e^{-iH}$, where the Hamiltonian $H_i=H_i^\dagger$ generates the gate $\hat{U}_i$.
When the above circuit is implemented on a real quantum computer, it will experience errors.
In this paper, we focus on \emph{coherent control errors} for which the perturbed circuit
\begin{align}\label{eq:qc_noise}
    \ket{\psi(\varepsilon)}=U_1(\varepsilon_1)\cdots U_N(\varepsilon_N)\ket{\psi_0}
\end{align}
shown in Figure~\ref{fig:circuit_ideal_noisy} (b) is executed, where the noisy gates $U_i(\varepsilon_i)$ take the form
\begin{align}\label{eq:qc_noise_Ui}
U_i(\varepsilon_i)=e^{-i(1+\varepsilon_i)H_i}.
\end{align}
Here, $\varepsilon_i$ are the components of a real-valued noise vector $\varepsilon\in\bbr^N$.
The noisy quantum circuit~\eqref{eq:qc_noise} is related to the noise-free one~\eqref{eq:qc} by setting the noise to zero, i.e., $U_i(0)=\hat{U}_i$ for $i=1,\dots N$ such that $\ket{\psi(0)}=\ket{\hat{\psi}}$.

Throughout the paper, we assume that the size of the noise is bounded by some scalar $\bar{\varepsilon}>0$, i.e.,
$|\varepsilon_i|\leq\bar{\varepsilon}$ for $i=1,\dots,N$ or, equivalently, $\lVert \varepsilon\rVert_{\infty}\leq\bar{\varepsilon}$.
Beyond the existence of such a bound, we make no further assumptions.
In particular, we do \emph{not} assume that
the errors are identical, i.e., that $\varepsilon_i=\varepsilon_j$ for $i\neq j$, or that the bound $\bar{\varepsilon}$ is sufficiently small (both are common assumptions, e.g., in the
literature on composite pulses~\cite{levitt1986composite,jones2003robust,merrill2014progress}).
We also do not assume that the coherent control errors follow certain statistics as, e.g., assumed by~\cite{skolik2022robustness} in the context
of quantum reinforcement learning.

In this section, we derive worst-case fidelity bounds for the noisy quantum circuit~\eqref{eq:qc_noise} based on Lipschitz bounds.
To this end, we first show in Section~\ref{subsec:lipschitz_fidelity} that a Lipschitz bound of the map $\varepsilon\mapsto\ket{\psi(\varepsilon)}$ implies a worst-case fidelity bound.
Next, we derive Lipschitz bounds based on the norms of the Hamiltonians (Section~\ref{subsec:lipschitz_simple})
and based on the coupling between subsequent gates (Section~\ref{subsec:lipschitz_pair}).

\subsection{Worst-case fidelity bounds via Lipschitz bounds}\label{subsec:lipschitz_fidelity}
Consider the following problem of bounding the worst-case fidelity of~\eqref{eq:qc_noise}~\cite{nielsen2011quantum}.
\begin{problem}\label{prob:worst_case_fidelity}
Given $\bar{\varepsilon}\geq0$, find $M(\bar{\varepsilon})\geq0$ such that, for any $\varepsilon\in\bbr^N$ with $\lVert \varepsilon\rVert_{\infty}\leq\bar{\varepsilon}$ and any initial state $\ket{\psi_0}$, it holds that
\begin{align}\label{eq:prob_worst_case_fidelity}
    |\braket{\psi(\varepsilon)|\hat{\psi}}|\geq 1-M(\bar{\varepsilon}).
\end{align}
\end{problem}
The bound~\eqref{eq:prob_worst_case_fidelity} quantifies the loss of fidelity due to the noise.
Note that it is a \emph{worst-case bound} both w.r.t.\ the coherent control error (it holds for any $\varepsilon\in\bbr^N$ with $\lVert \varepsilon\rVert_{\infty}\leq\bar{\varepsilon}$)
and w.r.t.\ the initial state (it holds for any initial state $\ket{\psi_0}$).
In the following, we show that~\eqref{eq:prob_worst_case_fidelity} holds when choosing $M$ based on a Lipschitz bound of $\ket{\psi(\varepsilon)}$.
\begin{definition}\label{def:lipschitz}
    A scalar $L>0$ is a \emph{Lipschitz bound} of $\ket{\psi}$ if
    \begin{align}\label{eq:def_lipschitz}
    \lVert\ket{\psi(\varepsilon)}-\ket{\psi(\varepsilon')}\rVert_2\leq L\lVert \varepsilon-\varepsilon'\rVert_{\infty}
    \end{align}
    for all $\varepsilon,\varepsilon'\in\bbr^N$~\footnote{We use an $\infty$-norm for the noise rather than a $2$-norm in order to obtain tighter bounds 
    in our main technical results.}.
\end{definition}
By definition, a Lipschitz bound $L$ bounds the worst-case amplification of the perturbation $\varepsilon$ on the resulting quantum state $\ket{\psi(\varepsilon)}$.
Lipschitz bounds are closely connected to the \emph{diamond distance}~\cite{kitaev1997quantum}.
This connection and its implications are discussed in more detail in Section~\ref{sec:fault_tolerance}.

The following result shows that a Lipschitz bound can be used to derive a worst-case fidelity bound as in Problem~\ref{prob:worst_case_fidelity}.
\begin{theorem}\label{thm:worst_case_fidelity}
    Suppose $L$ is a Lipschitz bound of $\ket{\psi}$.
    Then, for any $\varepsilon\in\bbr^N$ satisfying $\lVert \varepsilon\rVert_{\infty}\leq\bar{\varepsilon}$ for some 
    $\bar{\varepsilon}\geq0$ and any initial state $\ket{\psi_0}$, it holds that
    \begin{align}\label{eq:thm_worst_case_fidelity}
        |\braket{\psi(\varepsilon)|\hat{\psi}}|\geq 1-\frac{L^2\bar{\varepsilon}^2}{2}.
    \end{align}
\end{theorem}
\begin{proof}
Note that
\begin{align}\label{eq:thm_worst_case_fidelity_proof1}
    &2-2\,|\braket{\psi(\varepsilon)|\hat{\psi}}|\\\nonumber
    \leq&2-2\,\mathrm{Re}\{\braket{\psi(\varepsilon)|\hat{\psi}}\}\\\nonumber
    =&\braket{\psi(\varepsilon)|\psi(\varepsilon)}-\braket{\psi(\varepsilon)|\hat{\psi}}-\braket{\hat{\psi}|\psi(\varepsilon)}+\braket{\hat{\psi}|\hat{\psi}}\\\nonumber
    =&\lVert\ket{\psi(\varepsilon)}-\ket{\hat{\psi}}\rVert_2^2.
\end{align}
Since $\ket{\hat{\psi}}=\ket{\psi(0)}$, we can apply~\eqref{eq:def_lipschitz} with $\varepsilon'=0$ to infer
\begin{align}
    &2-2|\braket{\psi(\varepsilon)|\hat{\psi}}|\leq L^2\lVert \varepsilon\rVert_{\infty}^2\leq L^2\bar{\varepsilon}^2,
\end{align}
where we use $\lVert \varepsilon\rVert_{\infty}\leq \bar{\varepsilon}$ for the last step.
Rearranging the terms leads to~\eqref{eq:thm_worst_case_fidelity}.
\end{proof}

Thus, if a Lipschitz bound $L$ is known for the map $\varepsilon\mapsto\ket{\psi(\varepsilon)}$, then the worst-case loss of fidelity due to noise bounded by $\bar{\varepsilon}>0$
can be directly computed as $\frac{L^2\bar{\varepsilon}^2}{2}$.
The converse direction is equally important:
If we want to guarantee a worst-case fidelity of no less than $\calF$, i.e., $|\braket{\psi(\varepsilon)|\hat{\psi}}|
\geq\calF$ for any $\varepsilon\in\bbr^N$ 
with $\lVert \varepsilon\rVert_{\infty}\leq\bar{\varepsilon}$ and any initial state $\ket{\psi_0}$, then the error bound needs to be
sufficiently small in the sense that
\begin{align}
    \bar{\varepsilon}\leq\frac{\sqrt{2}}{L}\sqrt{1-\calF}.
\end{align}

\subsection{Norm-based Lipschitz bounds}\label{subsec:lipschitz_simple}
We have seen in the previous section that a Lipschitz bound of $\varepsilon\mapsto\ket{\psi(\varepsilon)}$ allows us to compute worst-case fidelity bounds.
With this motivation, we now turn to the problem of deriving Lipschitz bounds.
The following result shows that this can be done based on the norms of the Hamiltonians generating the gates.
\begin{theorem}\label{thm:lipschitz_unitary}
    For any $\varepsilon,\varepsilon'\in\bbr^N$ and any initial state $\ket{\psi_0}$, we have
\begin{align}\label{eq:thm_lipschitz_unitary}
\lVert \ket{\psi(\varepsilon)}-\ket{\psi(\varepsilon')}\rVert_2&\leq \sum_{i=1}^N\lVert H_i\rVert_2|\varepsilon_i-\varepsilon_i'|.
\end{align}
In particular, $\sum_{i=1}^N\lVert H_i\rVert_2$ is a Lipschitz bound of $\ket{\psi}$ and, for any $\varepsilon\in\bbr^N$ with $\lVert \varepsilon\rVert_{\infty}\leq\bar{\varepsilon}$ and any initial state $\ket{\psi_0}$, it holds that
\begin{align}\label{eq:thm_lipschitz_unitary_fidelity}
    |\braket{\psi(\varepsilon)|\hat{\psi}}|\geq 1-\left(\sum_{i=1}^N\lVert H_i\rVert_2\right)^2\frac{\bar{\varepsilon}^2}{2}.
\end{align}
\end{theorem}

The proof of Theorem~\ref{thm:lipschitz_unitary} can be found in Appendix~\ref{app:proof_lipschitz}.
According to the result, the fidelity loss induced by coherent control errors can be bounded using two ingredients:
the norms of the involved Hamiltonians $\lVert H_i\rVert_2$ and the noise bound $\bar{\varepsilon}$.
In particular, the robustness of the circuit improves if the norms $\lVert H_i\rVert_2$ decrease, and vice versa.

Intuitively, this can be explained by considering single-qubit rotations.
Suppose the circuit consists only of one $Z$-rotation with angle $\theta$, i.e., 
$\hat{U}=R_{\rmz}(\theta)=e^{-i\frac{\theta}{2} Z}$.
When affected by a coherent control error, the perturbed circuit takes the form
$U(\varepsilon)=e^{-i(1+\varepsilon)\frac{\theta}{2}Z}$ with $\varepsilon\in\bbr$.
Applying~\eqref{eq:thm_lipschitz_unitary_fidelity}, we directly infer
\begin{align}\label{eq:lipschitz_unitary_fidelity_example}
    |\braket{\psi(\varepsilon)|\hat{\psi}}|\geq 1-\frac{\theta^2 \varepsilon^2}{8}.
\end{align}
This formula illustrates the intuitive fact that,
due to the multiplicative influence of the noise, the sensitivity of the fidelity w.r.t.\ the noise 
decreases for smaller values of $\theta$, and it
vanishes for $\theta=0$.

Let us apply Theorem~\ref{thm:lipschitz_unitary} to the example given in Section~\ref{sec:intro} in order to compare the robustness
of the two circuits $\hat{U}=R_{\rmz}(\frac{\pi}{4})R_{\rmy}(\frac{\pi}{2})$ and 
$\hat{U}'=R_{\rmz}(-\frac{3\pi}{4})R_{\rmy}(-\frac{\pi}{2})$.
By computing the norms of all involved Hamiltonians and using Theorem~\ref{thm:lipschitz_unitary}, we infer that the Lipschitz bound of
$U$ w.r.t.\ the coherent control error is given by $\frac{3\pi}{8}$, whereas that of $U'$ is given by $\frac{5\pi}{8}$.
Hence, the worst-case fidelity bound~\eqref{eq:thm_lipschitz_unitary_fidelity} of $U'$ is smaller than that of $U$, which explains why
the latter is more robust w.r.t.\ coherent control errors, as observed in Section~\ref{sec:intro}.
Also note that the difference of the Lipschitz bound is due to the different magnitude of the $z$-rotations,
which is why the higher sensitivity of $U'$ only manifests itself as an over-/under-rotation around the $z$-axis, cf.\ Figure~\ref{fig:bloch_H}.

Theorem~\ref{thm:lipschitz_unitary}
is not only applicable to single-qubit operations but also to multi-qubit gates.
In particular, the bounds~\eqref{eq:thm_lipschitz_unitary} and~\eqref{eq:thm_lipschitz_unitary_fidelity} can be easily computed even for large, possibly high-dimensional systems, since 
the Hamiltonians $H_i$ are typically sparse, cf.\ Appendix~\ref{app:scalability} for details.

\begin{remark}\label{rk:global_phase}
    Since global phases are unobservable in the context of quantum computing, the Hamiltonian $H$ which generates a given unitary matrix $\hat{U}=e^{-iH}$ is not unique.
    More precisely, there is a family of Hamiltonians $H_{\varphi}=H+\varphi I$ with parameter $\varphi\in\bbr$ whose elements all produce the same unitary, modulo the global phase $\varphi$:
    \begin{align}
        e^{-iH_{\varphi}}=e^{-i\varphi}e^{-iH}.
    \end{align}
    As expected, the effect of coherent control errors on the circuit is independent of the particular choice of $H_{\varphi}$ since
    \begin{align}
        e^{-iH_{\varphi}(1+\varepsilon)}=e^{-i\varphi(1+\varepsilon)}e^{-iH(1+\varepsilon)}.
    \end{align}
    In particular, the influence of the error $\varepsilon$ on the circuit with Hamiltonian $H_{\varphi}$ is identical to the influence on $H$, modulo the global phase $\varphi(1+\varepsilon)$.
    Nevertheless, it is important to note that the theoretical bounds derived in this paper (Theorem~\ref{thm:lipschitz_unitary} as well as Theorem~\ref{thm:lipschitz_pair}
    in the following section) depend on the specific choice of $\varphi$.
    Clearly, $\lVert H\rVert_2\neq \lVert H+\varphi I\rVert_2$ for any $\varphi\neq0$.
   Theorem~\ref{thm:lipschitz_unitary} as well as all further statements in this paper remain true when the Hamiltonians $H_i$ are replaced by $H_i+\varphi_i I$ for any $\varphi_i\in\bbr$.
    In particular, the tightest Lipschitz bound of $\ket{\psi}$ (and, hence, the tightest worst-case fidelity bound) is achieved when replacing $H_i$ by $H_i+\varphi_i^*$, where
    $\varphi_i^*$ minimizes
    \begin{align}
        \min_{\varphi_i\in\bbr}\lVert H_i+\varphi_iI\rVert_2.
    \end{align}
\end{remark}

\begin{remark}\label{rk:related_bounds}
    A Lipschitz bound of the form $\sqrt{N}\max_i\lVert H_i\rVert_2$ was obtained in~\cite{sweke2020stochastic} for the purpose of convergence analysis of VQAs.
    In the present paper, we work with the bound derived in Theorem~\ref{thm:lipschitz_unitary} since it explicitly involves all elements of the quantum circuit which, as we will see later in the paper, 
    provides a flexible basis for robustness analysis.
    Further, we note that the results in~\cite{lidar2008distance} imply a norm-based Lipschitz bound analogous to~\eqref{eq:thm_lipschitz_unitary} for the case of one gate.
\end{remark}

An interesting question that arises is how tight the bounds in Theorem~\ref{thm:lipschitz_unitary} are,
i.e., how close they are to the smallest possible Lipschitz bound - the Lipschitz constant of $\ket{\psi}$ - and to the true worst-case fidelity.
In Appendix~\ref{app:tightness}, we show that the worst-case fidelity bound~\eqref{eq:thm_lipschitz_unitary_fidelity} is, for $N=1$, a good approximation
of the true worst-case fidelity, and that it is more accurate than alternative bounds.

It is important to note that the Lipschitz bound~\eqref{eq:thm_lipschitz_unitary} itself is, in general, not tight, e.g.,
since the coupling between sequentially applied gates is not taken into account.
In the next subsection, we derive structured Lipschitz bounds which involve pairs of sequentially applied Hamiltonians in order to
derive possibly tighter bounds.

\subsection{Pair-wise Lipschitz bounds}\label{subsec:lipschitz_pair}
The following result states Lipschitz bounds of $\ket{\psi}$ which take the coupling between subsequent gates into account.

\begin{theorem}\label{thm:lipschitz_pair}
    \begin{subequations}\label{eq:thm_lipschitz_pair}
    The following are Lipschitz bounds of $\ket{\psi}$:
    \begin{itemize}
        \item If the number of unitaries $N$ is even:
        \begin{align}\label{eq:thm_lipschitz_pair_even}
            &\sqrt{2}\sum_{i=1}^{\frac{N}{2}}\lVert \begin{pmatrix} H_{2i-1}&H_{2i}\end{pmatrix}\rVert_2.
        \end{align}

        \item If the number of unitaries $N$ is odd:
        \begin{align}
            \label{eq:thm_lipschitz_pair_odd}
            &\lVert H_N\rVert_2+\sqrt{2}\sum_{i=1}^{\frac{N-1}{2}}\lVert \begin{pmatrix}H_{2i-1}&H_{2i}\end{pmatrix}\rVert_2.
        \end{align}
    \end{itemize}

\end{subequations}
\end{theorem}
The proof of Theorem~\ref{thm:lipschitz_pair} is provided in Appendix~\ref{app:proof_pair}.
First, we note that the bound can be easily modified to a sum of arbitrary $1$- and $2$-pairs of $\lVert H_i\rVert_2$.
For example, in case that $N=4$, possible Lipschitz bounds of $\ket{\psi}$ are given by
\begin{align*}
&\sqrt{2}\,\lVert \begin{pmatrix}H_1&H_2\end{pmatrix}\rVert_2+\lVert H_3\rVert_2+\lVert H_4\rVert_2,\\
&\lVert H_1\rVert_2+\sqrt{2}\,\lVert\begin{pmatrix}H_2&H_3\end{pmatrix}\rVert_2+\lVert H_4\rVert_2,\>\text{etc.}
\end{align*}
All of these are valid Lipschitz bounds and, depending on the specific form of the Hamiltonians, they may all be different.

Further, combining Theorem~\ref{thm:lipschitz_pair} with Theorem~\ref{thm:worst_case_fidelity} (for the case of even $N$),
we obtain the worst-case fidelity bound
\begin{align}\label{eq:lipschitz_pair_fidelity}
    |\braket{\psi(\varepsilon)|\hat{\psi}}|\geq1-\left(\sum_{i=1}^N
    \lVert\begin{pmatrix}H_{2i-1}&H_{2i}
    \end{pmatrix}\rVert_2\right)^2
    \bar{\varepsilon}^2
\end{align}
for any $\varepsilon\in\bbr^N$ with $\lVert \varepsilon\rVert_{\infty}\leq\bar{\varepsilon}$.
This shows that the coupling between any two sequentially applied Hamiltonians influences the robustness against coherent control errors.

To study this point in more detail, let us analyze under what conditions the pair-wise bounds in Theorem~\ref{thm:lipschitz_pair} are beneficial.
For any two matrices $H_1$, $H_2$, it holds that
\begin{align}\label{eq:lipschitz_structured_triangle}
\lVert\begin{pmatrix}H_1&H_2\end{pmatrix}\rVert_2
=&\sqrt{\lambda_{\max}(H_1^\dagger H_1+H_2^\dagger H_2)}\\\nonumber
\leq&\sqrt{\lambda_{\max}(H_1^\dagger H_1)+\lambda_{\max}(H_2^\dagger H_2)}\\\nonumber
\leq&\sqrt{\lambda_{\max}(H_1^\dagger H_1)}+\sqrt{\lambda_{\max}(H_2^\dagger H_2)}\\\nonumber
=&\lVert H_1\rVert_2+\lVert H_2\rVert_2.
\end{align}
Inequality~\eqref{eq:lipschitz_structured_triangle} is strict in almost all cases.
To be precise, there are two main factors that contribute to a possible gap.
Let us start with the second inequality in~\eqref{eq:lipschitz_structured_triangle}, which uses $\sqrt{a+b}\leq\sqrt{a}+\sqrt{b}$ for arbitrary $a,b\geq0$.
This inequality is strict, i.e., $\sqrt{a+b}<\sqrt{a}+\sqrt{b}$, whenever $a,b>0$.
The best possible improvement, i.e., the maximum gap $\sqrt{2}$, occurs for $a=b$.
On the other hand, the first inequality in~\eqref{eq:lipschitz_structured_triangle} uses
\begin{align}\label{eq:noisy_Hamiltonians_interaction_gap}
&\lambda_{\max}(H_1^\dagger H_1+H_2^\dagger H_2)\\\nonumber
\leq&\lambda_{\max}(H_1^\dagger H_1)+\lambda_{\max}(H_2^\dagger H_2).
\end{align}
This inequality is strict if the eigenvectors corresponding to the maximum eigenvalues of $H_1^\dagger H_1$ and $H_2^\dagger H_2$ are linearly independent.
These findings reveal the following general principle for robustness:
the derived Lipschitz bound on the sequential application of two unitaries $U_{i+1}(\varepsilon_{i+1})U_i(\varepsilon_i)$ is not only influenced by the 
individual norms of the Hamiltonians $\lVert H_{i}\rVert_2+\lVert H_{i+1}\rVert_2$
but also by the structured, pair-wise norm bounds $\lVert\begin{pmatrix}H_{i}&H_{i+1}\end{pmatrix}\rVert_2$.
In particular, the fidelity loss is reduced if the singular vectors corresponding to the maximum singular values of $H_{i+1}$ and $H_i$ are linearly independent.
It is important to emphasize that this is the case \emph{even if the noise terms $\varepsilon_i$ entering the different unitaries are independent.}
In the remainder of the paper, we mainly focus on the norm-based bounds in Theorem~\ref{thm:lipschitz_unitary},
but drawing analogous conclusions using the pair-wise bounds in Theorem~\ref{thm:lipschitz_pair} is
an interesting issue for future research.

\section{Worst-case bounds for fault-tolerant quantum computing}\label{sec:fault_tolerance}

The threshold theorems provide fundamental bounds on the required accuracy for each gate to achieve fault-tolerant quantum computation~\cite{kitaev1997quantum,aharonov1997fault,aharonov2008fault,nielsen2011quantum}.
The required accuracy is typically quantified via the \emph{diamond distance}~\cite{kitaev1997quantum} of an error $\calE$ w.r.t.\ the identity operator $I$, which is defined as 
\begin{align}\label{eq:diamond_definition}
    \lVert\calE-I\rVert_{\diamond}=\sup_{\psi}\lVert (\calE\otimes I_d-I_{d^2})(\psi)\rVert_1.    
\end{align} 
Here, $d$ denotes the dimension of the underlying system on which $\calE$ acts.
In order to guarantee fault-tolerant quantum computation, it is required that $\lVert\calE-I\rVert_{\diamond}\leq\eta$ for some sufficiently small $\eta>0$~\cite{aharonov2008fault}.

Suppose now that, as previously in the paper, $\calE$ takes the form of a coherent control error, i.e., $\calE(\varepsilon)=e^{-i\varepsilon H}$ with $\varepsilon\in\bbr$ and some Hamiltonian $H=H^\dagger$.
It is clear that any Lipschitz bound $L$ for the map $\varepsilon\mapsto\calE(\varepsilon)$ also yields an upper bound on the diamond distance in~\eqref{eq:diamond_definition}, i.e.,
\begin{align}\label{eq:diamond_bound}
    &\lVert\calE(\varepsilon)-I\rVert_{\diamond}\\\nonumber 
    =&\sup_{\psi}\lVert (\calE(\varepsilon)\otimes I_d-I_{d^2})(\psi)\rVert_1\\\nonumber
=&\sup_{\psi}\lVert ((\calE(\varepsilon)-\calE(0))\otimes I_d)(\psi)\rVert_1\\\nonumber
\leq&\lVert\calE(\varepsilon)-\calE(0)\rVert_1\leq\sqrt{d}\lVert\calE(\varepsilon)-\calE(0)\rVert_2\\\nonumber
\leq&\sqrt{d}L\lVert \varepsilon\rVert_{\infty}.
\end{align}
Thus, the results from Section~\ref{sec:fidelity} provide explicitly computable and insightful bounds on the diamond distance for coherent control errors.

Let us discuss the connection of this insight to existing results.
First, we note that determining the diamond distance w.r.t.\ the identity experimentally can be challenging~\cite{blume2017demonstration}.
Therefore, different approaches to computing or estimating the diamond distance have been proposed.
For example, one can first determine the (average) fidelity, see, e.g.,~\cite{nielsen2002simple,thomas2011robustness}, and then use it to bound the diamond 
distance for coherent errors~\cite{beigi2011simplified,wallman2014randomized,wallman2015error,wallman2015estimating,kueng2016comparing,sanders2016bounding}.
Further diamond distance bounds in the presence of coherent errors when using randomized compiling are derived in~\cite{wallman2016noise}.
In contrast to these works, our framework provides tailored bounds for coherent control errors, which can be explicitly computed, depend directly
on the circuit elements, and hold under rather general assumptions, e.g., without requirements on statistics or the size of the error as well as for independent errors.
Finally, we note that Theorem~\ref{thm:lipschitz_unitary} also applies to scenarios with 
many gates, for which our bounds remain easily computable (cf.\ Appendix~\ref{app:scalability}).

\section{A new guideline for robust quantum algorithm design and transpilation}\label{sec:guidelines}

In this section, we discuss how the theoretical results in Section~\ref{sec:fidelity} can be used to  
derive a systematic and flexible guideline for designing and transpiling quantum algorithms such that their resilience against coherent control errors is improved.

Since noise poses a major obstacle to the practical demonstration of a quantum advantage, 
QEC~\cite{shor1995scheme,calderbank1996good,steane1996error}  and quantum error mitigation (QEM, see~\cite{cai2022quantum} for a recent survey) have been developed to handle noise.
In QEC, the ideal, noise-free circuit is expanded via additional circuit elements which can detect and/or compensate possible errors.
Although strong theoretical statements can be made on the success of such QEC approaches (cf.\ the discussion in Section~\ref{sec:fault_tolerance}), QEC can produce a significant overhead in terms of additional qubits and gates such that its practical application on current NISQ devices is challenging.
In QEM, instead of correcting for unavoidable errors on the circuit level, one instead leaves the quantum circuit unchanged and rather reduces the effect of noise via classical post-processing.
While QEM can bring practical advantages in the current NISQ era, there are also a number of open challenges
and, in particular, fundamental limitations~\cite{takagi2022fundamental,quek2022exponentially}.

Since QEC and QEM currently cannot completely eliminate errors, there is need for additional circuit optimization during the design or transpilation of quantum algorithms.
Important metrics for characterizing robustness of quantum circuits are the depth, the gate count, or
the number of entangling gates of the circuit.
With this motivation, a variety of approaches has been developed for optimizing circuits, e.g., during the transpilation step, in order to 
reduce these quantities~\cite{maslov2008quantum,arabzadeh2010rule,nam2018automated,amy2019controlled,lee2019hybrid,duncan2020graph,foesel2021quantum,nagarajan2021quantum}.
Such methods can significantly reduce the complexity of quantum circuits and, therefore, 
improve robustness against noise and enable an easier implementation on NISQ hardware.

Theorem~\ref{thm:lipschitz_unitary} can be used to derive a novel, quantitative guideline for improving the robustness of quantum algorithms during the design or transpilation step. 
To be precise, inequality~\eqref{eq:thm_lipschitz_unitary_fidelity} shows that the worst-case loss of fidelity due to coherent control errors depends on the norms of the Hamiltonians $H_i$ defining the quantum circuit.
Thus, smaller norms $\lVert H_i\rVert_2$ imply better robustness against such errors.
This means that, whenever an algorithm designer or a transpilation procedure has the choice between different Hamiltonians, choosing a circuit with smaller norms will improve robustness.

The circuit in Section~\ref{sec:intro} provides a concrete example demonstrating how this guideline can be used for algorithm design.
We have seen in simulation that $\hat{U}=R_{\rmz}(\frac{\pi}{4})R_{\rmy}(\frac{\pi}{2})$ is more robust against coherent control errors in comparison to
$\hat{U}'=R_{\rmz}(-\frac{3\pi}{4})R_{\rmy}(-\frac{\pi}{2})$.
This observation can be explained by our framework by noting that the derived Lipschitz bound of the noisy version $U(\varepsilon)$ of $\hat{U}$ against coherent control errors
is larger than that of $U'(\varepsilon')$ (cf.\ the discussion below Theorem~\ref{thm:lipschitz_unitary}).
Therefore, we can conclude that it is always beneficial (for robustness against coherent control errors) to implement $\hat{U}$ instead of $\hat{U}'$.

It is important to note that this conclusion cannot be drawn from existing circuit optimization schemes mentioned above:
both the circuit depth and the gate count for the two circuits are identical.
Yet, the loss of fidelity of $U(\varepsilon)$ is more than twice as large as that of $U'(\varepsilon)$.
This shows that Lipschitz bounds provide a more accurate, quantitative metric for assessing robustness of quantum algorithms in the presence of coherent control errors.
Certainly, for the simple example in Section~\ref{sec:intro}, one can arrive at the same conclusion without using Theorem~\ref{thm:lipschitz_unitary}, e.g., 
via the simulations shown in Figure~\ref{fig:bloch_H} or simply from intuition.
However, for more intricate scenarios, it may not be immediately obvious which gate sequence provides the most robust solution.
The proposed framework, on the other hand, can still be used to assess robustness.
Especially for large algorithms, it is important that the circuit is as robust as possible since, otherwise, the errors may quickly accumulate
and prevent a reliable execution.

To give another example, suppose we want to implement a non-trivial quantum algorithm and we have different realizations of this algorithm or different universal gate sets at our disposal.
We can then find the most robust implementation by simply computing a Lipschitz bound for each configuration as the sum of the norms of all $N$ involved Hamiltonians.
This can be done with low computational cost even for very large circuits, compare Appendix~\ref{app:scalability}.
After computing the Lipschitz bound for each circuit, the one with the smallest result has the best resilience against coherent control errors according to Theorem~\ref{thm:lipschitz_unitary}.
In Section~\ref{sec:robustness_elementary}, we follow this idea
by studying the robustness of the $3$-qubit QFT when transpiled into different elementary gate sets.

\section{Robustness of the Quantum Fourier Transform for different elementary gate sets}\label{sec:robustness_elementary}
In the following, we illustrate the practical potential of the proposed theoretical framework by solving the following problem:
we study the robustness of different elementary gate set implementations of the $3$-qubit QFT.
To this end, we consider the following gate sets:
\begin{itemize}
    \item Gate set A:\\
    $\sqrt{X}$, $X$, $R_{\rmz}$, $CX$\\
    (used by IBM~\cite{ibm_gateset})

    \item Gate set B:\\
    $R_{\rmx}(\pm \frac{\pi}{2})$, $R_{\rmx}(\pm \pi)$, $R_{\rmz}$, $CZ$\\
    (used by Rigetti~\cite{rigetti_gateset})

    \item Gate set C:\\
    $U_1$, $U_2$, $U_3$, $CX$\\
    (formerly used by IBM~\cite{ibm_gateset_old})

    \item Gate set D:\\
    $\sqrt{i\text{SWAP}}$, FSIM, PhasedXZ, $X$, $Y$, $Z$\\
    (used by Google~\cite{google_gateset})

    \item Gate set E:\\
     $R_{\rmx \rmy}(\frac{\pi}{2})$, $R_{\rmx \rmy}(\pi)$, $R_{\rmz}$, $U_{\rmz\rmz}$\\
     (used by Honeywell~\cite{honeywell_gateset})
\end{itemize}
In the following, we transpile the quantum circuit for the $3$-qubit QFT for each of these five gate sets and study the robustness w.r.t.\ coherent control errors
both theoretically and in simulation.
The transpilation is carried out using the Berkeley Quantum Synthesis Toolkit (BQSKit)~\cite{bqskit_software}, for which we use the maximum optimization level.
The textbook circuit for the $3$-qubit QFT as well as the five transpiled circuits according to the gate sets A--E are provided in Appendix~\ref{app:qft}.

\begin{figure}
    \begin{center}
    \includegraphics[width=\columnwidth]{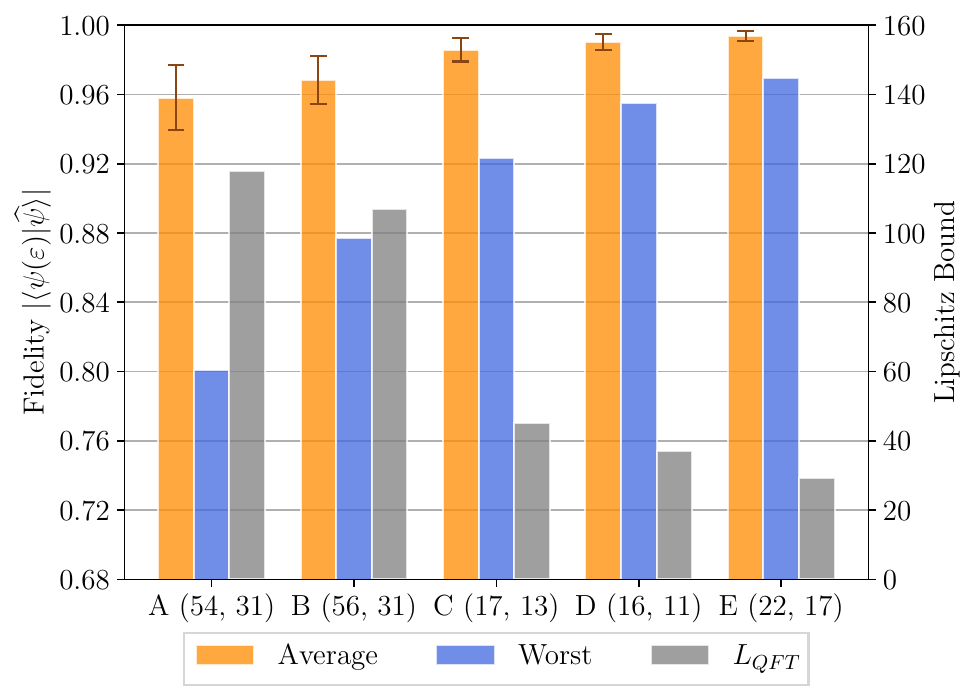}
    \end{center}
    \caption{
        Fidelities, Lipschitz bounds, gate counts, and circuit depths for five elementary gate set implementations of the $3$-qubit QFT, which are
        affected by coherent control errors.
        For each gate set, the figure shows the average fidelity including the standard deviation (left bar, orange),
        the worst-case fidelity (middle bar, blue),
        the Lipschitz bound (right bar, gray),
        and gate count (left number in parentheses) as well as the circuit depth (right number in parentheses).
        }
    \label{fig:experiment_gatesets}
\end{figure}

We start by computing the Lipschitz bound $L_{QFT}$ for each transpiled circuit based on Theorem~\ref{thm:lipschitz_unitary}.
The results are displayed in Figure~\ref{fig:experiment_gatesets}.
Note that gate set E leads to the smallest Lipschitz bound, followed by gate set D, etc.
Thus, using the worst-case fidelity bound~\eqref{eq:thm_lipschitz_unitary_fidelity}, 
we expect that the circuit corresponding to gate set E has the best robustness w.r.t.\ coherent control errors.
In the following, we validate this theoretical analysis in simulation.

Specifically, we simulate each circuit (QFT with elementary gate set A/B/C/D/E) in the presence of coherent control errors, i.e.,
we replace each noise-free gate $e^{-iH_i}$ appearing in each circuit by $e^{-i(1+\varepsilon_i)H_i}$ 
with independent noise terms $\varepsilon_i$, which are uniformly sampled
from $[-0.05,+0.05]$.
Furthermore, to estimate the worst-case fidelities, we draw the initial states $\ket{\psi_0}$ uniformly at random from the Haar measure.
Figure~\ref{fig:experiment_gatesets} displays the worst-case fidelity, average fidelity,
and standard deviation based on $40,000$ runs of each algorithm with varying noise realizations.
Note that there is a strong correlation between a small Lipschitz bound and a resulting large fidelity (average and worst), which
confirms the above analysis.
In particular, gate set E leads to the largest worst-case fidelity, whereas gate set A, which has the largest Lipschitz bound, is
the least robust.

It is also interesting to compare the gate counts and the circuit depths of the above circuits, which are commonly employed robustness quantifiers in existing circuit optimization methods (cf.\ the discussion in Section~\ref{sec:guidelines}).
These are shown in Figure~\ref{fig:experiment_gatesets} as well.
Note that, for the above circuits, having fewer gates or lower depths does not necessarily correspond to better robustness.
For example, the loss of the worst-case fidelity for gate set C is more than two times as large as that for gate set E,
even though the implementation based on gate set E has almost $30\%$ more gates.
A similar observation can be made when comparing the circuit depths.
This contradicts the common philosophy that less gates imply better robustness, and it shows that
Lipschitz bounds are, indeed, more accurate quantifiers of the robustness against coherent control errors.

Finally, for the above circuits, the worst-case fidelity bound~\eqref{eq:thm_lipschitz_unitary_fidelity} is conservative and even negative for four out of five circuits. 
Nevertheless, as Figure~\ref{fig:experiment_gatesets} shows, the Lipschitz bounds still provide a simple and quantitative measure for comparing their robustness. 

To summarize, our theoretical results provide novel insights into robust quantum algorithm design and transpilation, allowing to compare the robustness
of different circuits and, thereby, allowing for an informed choice of the most robust implementation.
This is especially important for large quantum algorithms:
if a quantum algorithm should produce any form of quantum advantage, then it cannot be efficiently simulated classically and, therefore,
it is not tractable to make a statistical analysis as in Figure~\ref{fig:experiment_gatesets}.
On the other hand, computing the Lipschitz bound is very simple and easily scalable, cf. Appendix~\ref{app:scalability}.

It should be emphasized that our main contribution is \emph{not} a robustness comparison of the elementary gate sets themselves.
In particular, we do not claim that gate set E is generally preferable over the others.
For example, it is entirely possible that
there exists a circuit based on gate set D which implements the QFT and has better robustness than the one based on gate set E shown in Appendix~\ref{app:qft}.
Our main contribution is \emph{the mere possibility} to make a theoretical analysis as above for a given set of quantum circuits.
Without resorting to simulations or experiments, we can use Lipschitz bounds to give a priori guarantees on the robustness against coherent control errors.
Such insights provide a promising tool for informing algorithm designers or circuit optimization methods in order to improve robustness directly at the design or transpilation stage.

In order to show that these theoretical findings can indeed be turned into a measurable robustness advantage in practice, 
we compare the implementation of two circuits with different Lipschitz bounds on a real quantum computer in Section~\ref{sec:validation}.

\section{Validation on a quantum computer}\label{sec:validation}
In the following, we validate our theoretical findings in an implementation on the \textit{ibm\_nairobi} quantum computer~\cite{ibm_hardware}.
To this end, we consider the two circuits $\hat{U}_A$ and $\hat{U}_B$ depicted in Figures~\ref{fig:circuit_ideal_noisy}
(a) and (b), respectively. 
It can be readily verified that both circuits produce the same output state $\ket{\hat{\psi}}$.

\begin{figure}
    \centering 
    \begin{subfigure}{\columnwidth}
        \centering
        \includegraphics[width=0.7\textwidth]{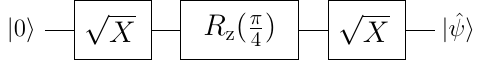}
        \caption{Ideal circuit $\hat{U}_A$}
    \end{subfigure}
\vskip5pt
\begin{subfigure}{\columnwidth}
    \centering
    \includegraphics[width=0.85\textwidth]{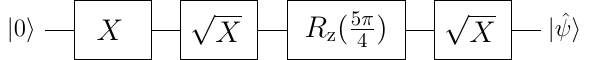}
    \caption{Ideal circuit $\hat{U}_B$}
\end{subfigure}
\vskip5pt
\begin{subfigure}{\columnwidth}
    \centering
    \includegraphics[width=0.85\textwidth]{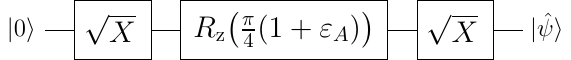}
    \caption{Noisy circuit $U_A(\varepsilon_A)$}
\end{subfigure}
\vskip5pt
    \begin{subfigure}{\columnwidth}
        \centering
        \includegraphics[width=\textwidth]{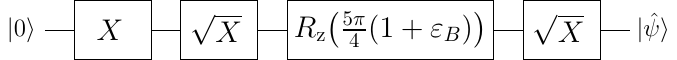}
        \caption{Noisy circuit $U_B(\varepsilon_B)$}
    \end{subfigure}
    \caption{Ideal circuits $U_A$ and $U_B$, and their noisy versions
    $U_A(\varepsilon_A)$ and $U_B(\varepsilon_B)$.)}	\label{fig:circuit_ideal_noisy}
\end{figure}

In the following, we use Theorem~\ref{thm:lipschitz_unitary} to show that $U_A$ is more robust against coherent control errors than $U_B$,
and we confirm this statement via
an implementation on a quantum computer.
To be precise, we consider the case that the $R_{\rmz}$-gates are affected by coherent control errors $\varepsilon_A$ and $\varepsilon_B$, respectively,
cf. Figures~\ref{fig:circuit_ideal_noisy} (c) and (d).
We do not consider coherent control errors of the $X$ and $\sqrt{X}$ gates since the 
elementary gate set which can be implemented on \textit{ibm\_nairobi} does not contain $R_{\rmx}$-rotations.

\begin{figure}
    \begin{center}
    \includegraphics[width=\columnwidth]{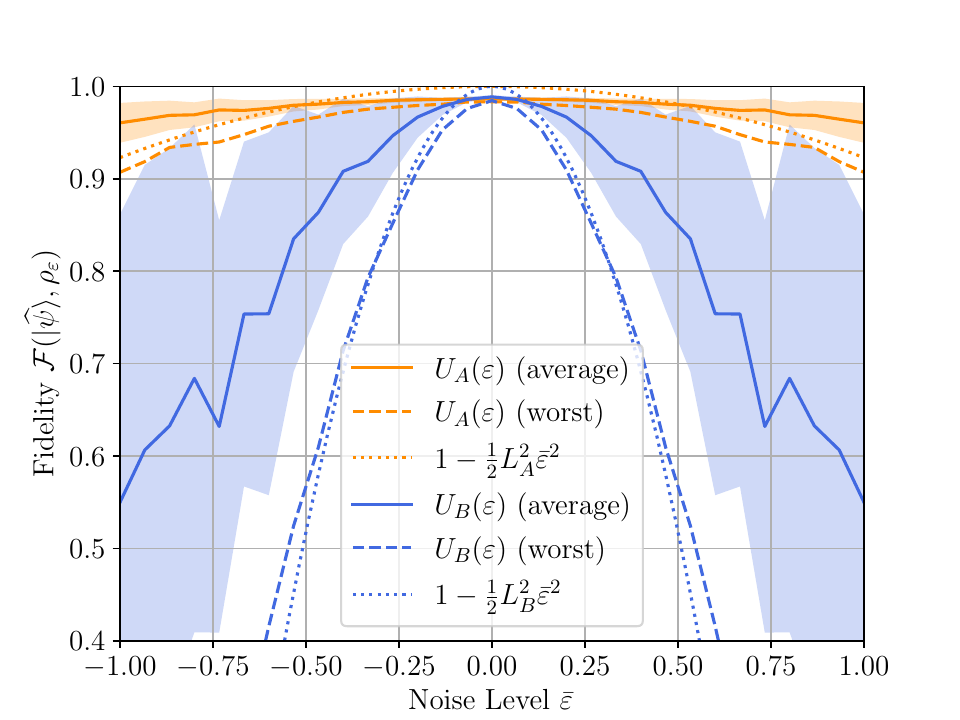}
    \end{center}
    \caption{
        Fidelities for circuits $U_A$ (orange) and $U_B$ (blue) depending on the noise level:
        average (solid), standard deviation (shaded area), and worst case (dashed) on \textit{ibm\_nairobi} over 80 noise samples for each noise level $\bar{\varepsilon}$, and
        worst-case bound from Theorem~\ref{thm:lipschitz_unitary} (dotted).
        }
    \label{fig:experiment_real}
    \end{figure}

Using Theorem~\ref{thm:lipschitz_unitary}, we can compute the Lipschitz bounds of $U_A$ and $U_B$
w.r.t.\ the errors $\varepsilon_A$ and $\varepsilon_B$ as
\begin{align}
    L_A=\frac{\pi}{8}\quad\text{and}\quad L_B=\frac{5\pi}{8},
\end{align}
respectively~\footnote{It is easy to show that the additional gates $X$ and $\sqrt{X}$ do not affect the Lipschitz bounds
w.r.t.\ the coherent control errors in $R_\rmz$.}.
Figure~\ref{fig:experiment_real} shows the resulting worst-case fidelity bounds according to~\eqref{eq:thm_lipschitz_unitary_fidelity},
where the bound for $U_A$ is substantially larger than that for $U_B$.
Thus, Theorem~\ref{thm:lipschitz_unitary} predicts that $U_A$ is more robust against coherent control errors than $U_B$.

For the implementation on \textit{ibm\_nairobi}, we run the circuits $U_A(\varepsilon_A)\ket0$ and $U_B(\varepsilon_B)\ket0$ for 16 equidistant noise levels $\bar{\varepsilon} \in [0, 1]$.
For each noise level, we draw 80 random samples of $\varepsilon_A$ and $\varepsilon_B$, respectively, from $[-\bar{\varepsilon},\bar{\varepsilon}]$.
Since the fidelity cannot be measured directly on a quantum computer, we apply a simple quantum state tomography (QST) procedure to estimate the underlying density matrix $\rho_{\varepsilon}$ and calculate the fidelity $\mathcal{F}(\ket{\hat{\psi}}, \rho_{\varepsilon})$ based on that.
Our employed QST procedure for a single qubit is outlined in Appendix~\ref{app:qst}.

Figure~\ref{fig:experiment_real} shows the resulting average and worst-case fidelities for both circuits depending on the noise level.
First, we observe that, for non-trivial noise levels, circuit $U_A$ indeed performs significantly better than circuit $U_B$, as predicted by Theorem~\ref{thm:lipschitz_unitary}.
Further, the gap not only increases for larger noise levels, but the observed worst-case fidelity is, in fact, very close to the
theoretical bound~\eqref{eq:thm_lipschitz_unitary_fidelity}.
For $U_A$, the empirical value is smaller than the theoretical lower bound, which can be explained by the inexact fidelity estimation procedure
and possibly further errors caused, e.g., by decoherence effects.
To conclude, is important to note that the improved robustness of $U_A$ over $U_B$ can only be explained by 
the different Lipschitz bounds and not, e.g., by the number of gates ($3$ vs.\ $4$) since
both circuits have almost identical fidelity in the absence of coherent control errors, i.e., for $\bar{\varepsilon}=0$.

Finally, further experimental support for the presented theoretical analysis is provided in~\cite{gokhale2020optimized}, which similarly shows that implementations with smaller Lipschitz bounds (realized via smaller/shorter pulses in the experimental setup) indeed lead to a higher fidelity.

\section{Variational quantum algorithms:
robustness via regularization}\label{sec:vqa}

Besides employing Lipschitz bounds for quantum circuit design and transpilation as outlined in the previous sections, they are also useful in VQAs. 
To be precise, we show in the following that they can be used to fine-tune VQAs towards outputting more robust quantum circuits.

VQAs are promising candidates for achieving a quantum advantage in the near-term future on NISQ devices~\cite{cerezo2021variational}.
They contain parametrized quantum circuits that are executed repeatedly and adapted via an optimization scheme.
More precisely, VQAs involve parametrized unitaries of the form
\begin{align}
    \hat{U}(\theta)=\hat{U}_1(\theta_1)\cdots \hat{U}_N(\theta_N),
\end{align}
where $\theta\in\bbr^N$ is a free parameter and $\hat{U}_i(\theta_i)=e^{-i\theta_i H_i}$ for Hamiltonians $H_i=H_i^\dagger$, $i=1,\dots,N$.
The idea is to vary $\theta$ in order to minimize a cost function $\calC:\bbr^N\to\bbr$, which is typically the measurement of some observable $\calM$ after applying $\hat{U}(\theta)$ to an initial state $\ket{\psi_0}$:
\begin{align}\label{eq:vqa_cost}
    \calC(\theta)=\braket{\psi_0|\hat{U}(\theta)^\dagger\calM \hat{U}(\theta)|\psi_0}.
\end{align}
Popular algorithms following this idea include, e.g., the quantum approximate optimization algorithm (QAOA)~\cite{farhi2014quantum}, the variational quantum eigensolver (VQE)~\cite{peruzzo2014variational}, and many more~\cite{bharti2022noisy}.
Such algorithms are expected to be more robust w.r.t.\ noise and implementable in the near future because they already produce meaningful results when using only few qubits and gates, enabling their implementation on currently available hardware.
    
Suppose now that the ideal unitaries $\hat{U}_i(\theta_i)$ are affected by coherent control errors $\varepsilon_i\in\bbr$, i.e., instead of $\hat{U}(\theta)$, we have
\begin{align}
    U(\theta,\varepsilon)=e^{-i\theta_1(1+\varepsilon_1)H_1}\cdots e^{-i\theta_N(1+\varepsilon_N)H_N}.
\end{align}
Theorem~\ref{thm:lipschitz_unitary} now implies that, for a fixed parameter $\theta\in\bbr^N$, $\sum_{i=1}^N|\theta_i|\lVert H_i\rVert_2$ is a Lipschitz bound of $\varepsilon\mapsto U(\theta,\varepsilon)\ket{\psi_0}$.
In particular, smaller values of $|\theta_i|$ imply better robustness of the VQA against coherent control errors.
This motivates solving the following regularized optimization problem in order to keep the size of $\theta$ preferably small.
\begin{align}
    \min_{\theta\in\bbr^N}\calC(\theta)+\lambda \lVert\theta\rVert_2^2.
\end{align}
Here, $\lambda>0$ is a parameter which can be used to trade off optimality and robustness against coherent control errors.
For larger values of $\lambda$, the solution of the VQA will tend more strongly towards a solution which provides a robust algorithm, whereas
smaller values of $\lambda$ will encourage better performance at the price of possibly worse robustness.
We note that, for VQAs, there exist works suggesting regularization of Hamiltonian parameters~\cite{park2020practical} and studying coherent control errors~\cite{sung2020using,ito2021universal,skolik2022robustness,rabinovich2023gate}.
However, the above-described link between regularization and coherent control errors is, to the best of our knowledge, new.

Robustness of VQAs is especially important if the algorithm should be transferrable between different quantum devices.
Adding a regularization in the VQA improves robustness and, therefore, it encourages an optimal solution which is less dependent on the specific hardware on which the VQA is trained.

In the following, we confirm the above discussion with an example, showing that, indeed, regularization of $\theta$ leads to more robust quantum circuits.
For this, we learn a simple quantum model for a regression task.
That is, given access to a data set $\mathcal{D} = \{ (x, y)\}_{i=1}^n$ consisting of $n$ real-valued data points $x \in \mathcal{X}, y \in \mathcal{Y}$
we aim to train a model $\hat{f}_\theta$ that, upon input of $x_i$, 
outputs $\hat{f}_\theta(x_i) = \hat{y}_i$, which is close to $y_i$ in the sense that the mean squared error
\begin{align}\label{eq:mean_squared_error}
    \text{MSE}(\theta)
    =
    \frac{1}{n} \sum_{i=1}^n (y_i - \hat{f}_\theta(x_i))^2
\end{align}
is sufficiently small. 
Our model is based on~\cite{schuld2021effect} and consists of a single qubit.
To be more precise, we employ the unitary
\begin{align}
    U(\theta, x) = R_{\rmx}(x) R_{\rmz}(\theta_1) R_{\rmy}(\theta_2)R_{\rmz}(\theta_3)
\end{align}
as ansatz.
Note that this ansatz consists of one fixed gate $R_{\rmx}$ to encode the input $x$ and three trainable parameters $\theta = (\theta_1, \theta_2, \theta_3)$.
For a given input $x_i$ and parameter set $\theta$, we infer the output according to
\begin{align}\label{eq:inference_qml_model}
    \hat{f}_\theta(x_i) = \braket{\psi(\theta, x_i) | \hat{Z} | \psi(\theta, x_i)},
\end{align}
with $\ket{\psi(\theta, x_i)} = U(\theta, x_i) \ket{0}$, where we estimate this expectation value based on $20,000$ circuit evaluations.
For the training, we choose $\mathcal{C}(\theta) = \text{MSE}(\theta)$ as cost function and apply the parameter-shift rule~\cite{Schuld2019ParameterShiftRule} 
to calculate the gradient $\nabla \mathcal{C}(\theta)$.
The parameters are updated from step $i \rightarrow i+1$ based on the ADAM optimizer~\cite{kingdma2014adamOptimizer}, 
i.e., $\theta^{(i+1)} = \theta^{(i)} - \eta^{(i)} \nabla \mathcal{C}(\theta^{(i)})$, where $\eta^{(i)}$ is the chosen learning rate for step $i$.
It is shown in~\cite{schuld2021effect} that the above model can represent a sine-function $f(x)=\sin(x)$.
We create a simple training set $\mathcal{X}$ with chosing $n=20$ points equidistantly from $[0, 2\pi]$, 
and a target set $\mathcal{Y} = \{ \sin(x) : x \in \mathcal{X} \}$.

When evaluating the above quantum circuit, we assume that the $R_\rmy$- and $R_\rmz$-gates are affected by coherent control errors.
In particular, in each iteration, the rotation gates for the trainable parameters $R_{\rmy}$ and $R_{\rmz}$ are perturbed
as in~\eqref{eq:qc_noise_Ui} by noise terms sampled uniformly from $[-\bar{\varepsilon},\bar{\varepsilon}]$ with $\bar{\varepsilon}=0.05$.

We consider five different choices for the regularization parameter $\lambda\in\{0, 0.01, 0.05, 0.1, 0.5\}$, and, for each choice, learn eight models with different initial weights uniformly drawn from $[-2\pi,+2\pi]$. 
The final model is chosen after $50$ iterations of the ADAM optimizer as the one with the smallest cost $\calC(\theta)$. 

\begin{figure}
    \begin{center}
    \includegraphics[width=\columnwidth]{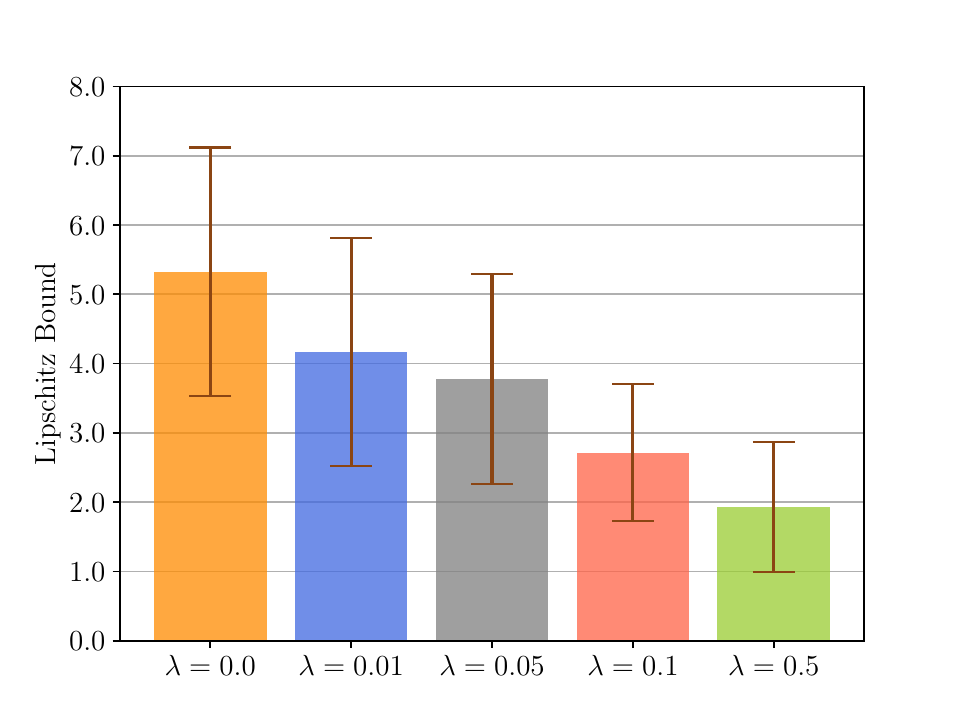}
    \end{center}
    \caption{
        Average and standard deviation of the Lipschitz bounds for the trained quantum circuits with varying regularization parameter $\lambda$.
         }
    \label{fig:experiment_vqa}
\end{figure}

The resulting Lipschitz bounds of the final circuits trained with different regularization parameters are shown in Figure~\ref{fig:experiment_vqa}.
We observe that a larger regularization parameter $\lambda$ indeed leads to a lower Lipschitz bound.
Following Theorem~\ref{thm:worst_case_fidelity}, this implies that these circuits are more robust w.r.t. coherent control errors.
Thus, we can conclude that the robustness of quantum circuits in such a setting can be successfully improved via parameter regularization.

Finally, in the context of VQAs, Lipschitz bounds have a further application beyond robustness against coherent control errors:
they can be used to certify convergence of VQA training.
It is well-known that, if the gradient of a cost function has a Lipschitz bound $L$, then gradient descent with
step-size less than $\frac{2}{L}$ provably converges~\cite{nesterov2004introductory}.
Using similar ideas,~\cite{sweke2020stochastic} study convergence of VQA optimization schemes in the presence of noise.
Our results in Section~\ref{sec:fidelity} provide Lipschitz bounds for the state $\ket{\psi}$.
The exact same Lipschitz bounds hold for a parametrized quantum circuit in the absence of coherent control errors, i.e., for
the map $\theta\mapsto\hat{U}(\theta)\ket{\psi_0}$.
These bounds then yield Lipschitz bounds for the gradient $\nabla\calC$ analogous to~\cite{sweke2020stochastic}.
Since the bounds in Theorems~\ref{thm:lipschitz_unitary} and~\ref{thm:lipschitz_pair} can be tighter than those in~\cite{sweke2020stochastic},
they can possibly allow for more degrees of freedom in the optimization algorithm and its convergence analysis,
thus reducing conservatism and improving convergence of the VQA training.

\section{Conclusion}\label{sec:conclusion}

In this paper, we presented a novel framework for analyzing robustness of quantum algorithms against coherent control errors.
We used Lipschitz bounds to derive worst-case fidelity bounds, which are explicitly computable and involve the norms of the Hamiltonians as well as their coupling.
Our framework was connected to several important problems in quantum computing:
We showed how the derived bounds can be used in threshold theorems for fault-tolerant quantum computing
and we provided a novel interpretation of parameter regularization in VQAs.
Moreover, we used our theoretical results to state a practical guideline for the design and transpilation of robust quantum algorithms,
which amounts to keeping the norms of the Hamiltonians small.
We showed that these norms quantify the robustness against coherent control errors more precisely than existing measures such as
the gate count by studying robustness of different elementary gate set implementations of the $3$-qubit QFT.
We demonstrated the practicality of our framework in multiple applications in simulation and on the \textit{ibm\_nairobi} quantum computer.
It will be an interesting next step to study robustness of further classes of quantum algorithms and gate sets against coherent control errors
in order to realize robust circuits which are more easily implementable on noisy quantum hardware.
In particular, developing circuit optimization or transpilation procedures which aim at reducing the norms of the Hamiltonians
is a promising avenue for improving robustness of quantum algorithms.
Finally, extending our theoretical analysis to include stochastic information about the errors may provide additional insights over the worst-case bounds considered in this paper.

\begin{acknowledgments}
    We would like to thank Elies Gil-Fuster and Johannes Jakob Meyer for their helpful feedback 
    on an earlier version of the manuscript.
    We acknowledge the use of IBM Quantum services for this work.
    The views expressed are those of the authors, and do not reflect the official policy or position of IBM or the IBM Quantum team.
    This work was funded by Deutsche Forschungsgemeinschaft (DFG, German Research Foundation) under Germany's Excellence Strategy - EXC 2075 - 390740016.
    We acknowledge the support by the Stuttgart Center for Simulation Science (SimTech).
\end{acknowledgments}

\appendix

\section{Proof of Theorem~\ref{thm:lipschitz_unitary}}\label{app:proof_lipschitz}

\begin{proof}
    For $i=1,\dots,N$, we have
    \begin{align*}
    \frac{\rmd}{\rmd \varepsilon_i}U_i(\varepsilon_i)=-iH_ie^{-i(1+\varepsilon_i)H_i}.
    \end{align*}
    In particular, the derivative of $U_i$ is uniformly bounded, i.e.,
    \begin{align*}
    \Big\lVert\frac{\rmd}{\rmd \varepsilon_i}U_i(\varepsilon_i)\Big\rVert_2
    \leq \lVert H_i\rVert_2
    \end{align*}
    for all $\varepsilon_i\in\bbr$.
    This implies that $\lVert H_i\rVert_2$ is a Lipschitz bound of $U_i$, i.e.,
    \begin{align}\label{eq:thm_lipschitz_unitary_proof}
    \left\lVert U_i(\varepsilon_i)-U_i(\varepsilon_i')\right\rVert_2\leq\lVert H_i\rVert_2|\varepsilon_i-\varepsilon_i'|
    \end{align}
    for all $\varepsilon_i,\varepsilon_i'\in\bbr$.
    Now, we infer
    \begin{align*}
    &\lVert\ket{\psi(\varepsilon)}-\ket{\psi(\varepsilon')}\rVert_2\\
    =&\Big\lVert \Big(\prod_{i=1}^NU_i(\varepsilon_i)-\prod_{i=1}^NU_i(\varepsilon_i')\Big)\ket{\psi_0}\Big\rVert_2\\
    \leq&\Big\lVert \prod_{i=1}^NU_i(\varepsilon_i)-\prod_{i=1}^NU_i(\varepsilon_i')\Big\rVert_2\\
    =&\Big\lVert U_1(\varepsilon_1)\Big(
    \prod_{i=2}^{N}U_i(\varepsilon_i)-\prod_{i=2}^{N}U_i(\varepsilon_i')\Big)\\
    &+(U_1(\varepsilon_1)-U_1(\varepsilon_1'))\prod_{i=2}^{N}U_i(\varepsilon_i')\Big\rVert_2\\
    \leq&\lVert U_1(\varepsilon_1)-U_1(\varepsilon_1')\rVert_2+
    \Big\lVert\prod_{i=2}^{N}U_i(\varepsilon_i)-\prod_{i=2}^{N}U_i(\varepsilon_i')\Big\rVert_2,
    \end{align*}
    where we use that $\ket{\psi_0}$ has unit norm and the $U_i$'s are unitary, 
    as well as the triangle inequality and sub-multiplicativity of matrix norms.
    Proceeding inductively, this implies
    \begin{align}\label{eq:thm_lipschitz_unitary_proof2}
    &\lVert\ket{\psi(\varepsilon)}-\ket{\psi(\varepsilon')}\rVert_2
    \leq\sum_{i=1}^N\lVert U_i(\varepsilon_i)-U_i(\varepsilon_i')\rVert_2.
    \end{align}
    Combining~\eqref{eq:thm_lipschitz_unitary_proof} and~\eqref{eq:thm_lipschitz_unitary_proof2}, we obtain
    \begin{align}\label{eq:thm_lipschitz_unitary_proof3}
        \lVert \ket{\psi(\varepsilon)}-\ket{\psi(\varepsilon')}\rVert_2&\leq\sum_{i=1}^N\lVert H_i\rVert_2| \varepsilon_i-\varepsilon_i'|\\\nonumber
        &\leq\sum_{i=1}^N\lVert H_i\rVert_2\lVert \varepsilon-\varepsilon'\rVert_{\infty},
        \end{align}
        which implies~\eqref{eq:thm_lipschitz_unitary} and the Lipschitz bound $\sum_{i=1}^N\lVert H_i\rVert_2$.
    Further, Theorem~\ref{thm:worst_case_fidelity} implies~\eqref{eq:thm_lipschitz_unitary_fidelity}, which thus concludes the proof.
    \end{proof}

\section{Scalability of Theorem~\ref{thm:lipschitz_unitary} for sparse circuits}\label{app:scalability}

In the following, we formulate the simple but practically important observation that the Lipschitz
bound~\eqref{eq:thm_lipschitz_unitary} only depends on the non-trivial gate components acting on the individual qubits.

\begin{proposition}\label{prop:individual_qubits}
    Suppose $\hat{U}_i$ only acts on a subset of all qubits, i.e.,
    \begin{align*}
        \hat{U}_i=e^{-iH_i}=I_{\otimes j_1}\otimes e^{-i\bar{H}_{i}}\otimes I_{\otimes j_2}
    \end{align*}
    for some $j_1$, $j_2$.
    Then, $\lVert H_i\rVert_2=\lVert \bar{H}_{i}\rVert_2$.
\end{proposition}
\begin{proof}
    Note that
\begin{align*}
    e^{-iH_i}=&\,I_{\otimes j_1}\otimes e^{-i\bar{H}_i}\otimes I_{\otimes j_2}\\\nonumber
    =&\,e^{I_{\otimes j_1}\otimes (-i\bar{H}_i)\otimes I_{\otimes j_2}}.
\end{align*}
This implies
\begin{align*}
\lVert H_i\rVert_2=
\lVert I_{\otimes j_1}\otimes \bar{H}_{i}\otimes I_{\otimes j_2}
\rVert_2=\lVert \bar{H}_{i}\rVert_2.
\end{align*}
\end{proof}

In the common scenario that the quantum algorithm only consists of single-qubit and $2$-qubit gates,
Proposition~\ref{prop:individual_qubits} implies that $\lVert H_i\rVert_2$ can be easily
determined by computing the norm of a $2\times 2$ or $4\times 4$ matrix.
In particular, the complexity of computing the bounds in Theorem~\ref{thm:lipschitz_unitary} for the full circuit grows only linearly with the
gate count and, therefore, the approach is easily scalable to large circuits.

\section{Tightness of the worst-case fidelity bound~\eqref{eq:thm_lipschitz_unitary_fidelity}}\label{app:tightness}
It is not hard to show that, for the case of a single gate ($N=1$) and for $\lVert H\rVert_2\bar{\varepsilon}\leq\frac{\pi}{2}$, the worst-case fidelity is equal to $|\cos(\lVert H\rVert_2\bar{\varepsilon})|$.
This means that there exist an initial state $\ket{\psi_0}$ and a noise realization $\varepsilon\in\bbr$ with $|\varepsilon|\leq\bar{\varepsilon}$ such that
\begin{align}\label{eq:tightness_cos}
    |\braket{\psi(\varepsilon)|\hat{\psi}}|=|\cos(\lVert H\rVert_2\bar{\varepsilon})|.
\end{align}
Hence, for $N=1$, the fidelity bound $1-\frac{1}{2}\lVert H\rVert_2^2\bar{\varepsilon}^2$ in~\eqref{eq:thm_lipschitz_unitary_fidelity}
is the second-order Taylor approximation of the exact worst-case fidelity.
This shows that the bound is not too conservative as long as the product $\lVert H\rVert_2\bar{\varepsilon}$ is small.

On the other hand, for multiple gates, i.e., $N>1$, finding an explicit expression for the worst-case fidelity is more involved.
As an alternative to the bound~\eqref{eq:thm_lipschitz_unitary_fidelity},
one can use~\eqref{eq:tightness_cos} to derive a lower bound as follows:
Recall that the fidelity $|\braket{\psi_1|\psi_2}|$ induces a metric 
$D(\ket{\psi_1},\ket{\psi_2})=\sqrt{1-|\braket{\psi_1|\psi_2}|^2}$~\cite{nielsen2011quantum}.
Thus, we can quantify the distance between the ideal state $\ket{\hat{\psi}}$ and the noisy state $\ket{\psi(\varepsilon)}$ by 
repeatedly applying the triangle inequality to this metric and using~\eqref{eq:tightness_cos}, i.e.,
\begin{align*}
    D(\ket{\psi(\varepsilon)},\ket{\hat{\psi}})\leq &\sum_{i=1}^N\sqrt{1-|\cos(\lVert H_i\rVert_2\bar{\varepsilon})|^2}\\
    &=\sum_{i=1}^N|\sin(\lVert H_i\rVert_2\bar{\varepsilon})|.
\end{align*}
This can be translated into the fidelity bound
\begin{align}\label{eq:fidelity_bound_sin}
    |\braket{\psi(\varepsilon)|\hat{\psi}}|
    \geq\sqrt{1-\left(\sum_{i=1}^N|\sin(\lVert H_i\rVert_2\bar{\varepsilon})|\right)^2},
\end{align}
assuming that the expression under the square root is positive.
Let us compare the bounds~\eqref{eq:thm_lipschitz_unitary_fidelity} and~\eqref{eq:fidelity_bound_sin}.
If the noise level $\bar{\varepsilon}$ is small, $\sin(\lVert H_i\rVert_2\bar{\varepsilon})$ in~\eqref{eq:fidelity_bound_sin} can be replaced by $\lVert H_i\rVert_2\bar{\varepsilon}$.
Defining $c=\sum_{i=1}^N\lVert H_i\rVert_2\bar{\varepsilon}$, the right-hand side of~\eqref{eq:fidelity_bound_sin} then becomes
$\sqrt{1-c^2}$.
    On the other hand, the right-hand side of~\eqref{eq:thm_lipschitz_unitary_fidelity} can be written as
    $1-\frac{c^2}{2}$.
    Since $\sqrt{1-c^2}<1-\frac{c^2}{2}$ for any $0<c<1$, we have thus proven that (for small noise levels) the 
    bound~\eqref{eq:thm_lipschitz_unitary_fidelity}
    is tighter than~\eqref{eq:fidelity_bound_sin},
i.e., it guarantees a larger worst-case fidelity.
Therefore, throughout the paper, we focus on the bound~\eqref{eq:thm_lipschitz_unitary_fidelity} derived in Theorem~\ref{thm:lipschitz_unitary}.

\section{Proof of Theorem~\ref{thm:lipschitz_pair}}\label{app:proof_pair}

\begin{proof}
    We only prove~\eqref{eq:thm_lipschitz_pair_even} and note that~\eqref{eq:thm_lipschitz_pair_odd} can be proven with trivial modifications.
    Note that
\begin{align}\label{eq:lipschitz_structured1}
&\frac{\rmd\ket{\psi(\varepsilon)}}{\rmd \varepsilon}\\\nonumber
=&-i\Big[ H_1\prod_{i=1}^NU_i(\varepsilon_i)\ket{\psi_0},\>U_1(\varepsilon_1)H_2\prod_{i=2}^NU_i(\varepsilon_i)\ket{\psi_0},\\\nonumber
&\qquad\dots,\prod_{i=1}^{N-1}U_i(\varepsilon_i)H_NU_N(\varepsilon_N)\ket{\psi_0}\Big].
\end{align}
Since $H_i$ and $U_i(\varepsilon_i)$ commute, we can rewrite the first two block-entries on the right-hand side of~\eqref{eq:lipschitz_structured1} as
\begin{align}\label{eq:lipschitz_structured_proof}
&\Big[ H_1\prod_{i=1}^NU_i(\varepsilon_i)\ket{\psi_0},U_1(\varepsilon_1)H_2\prod_{i=2}^NU_i(\varepsilon_i)\ket{\psi_0}\Big]\\\nonumber
&=U_1(\varepsilon_1)\begin{pmatrix}H_1&H_2\end{pmatrix}\Big(I_2\otimes \big(\prod_{i=2}^NU_i(\varepsilon_i)\big)\Big)\ket{\psi_0}.
\end{align}
Using that $\ket{\psi_0}$ is a unit vector and the matrices $U_i(\varepsilon_i)$ are unitary, the norm of~\eqref{eq:lipschitz_structured_proof}
is upper bounded by $\lVert\begin{pmatrix}H_1&H_2\end{pmatrix}\rVert_2$.
Thus similar to the proof of Theorem~\ref{thm:lipschitz_unitary}, we infer
that, for any $\varepsilon\in\bbr^N$,
\begin{align*}
    &\lVert\ket{\psi(\varepsilon)}-\ket{\hat{\psi}}\rVert_2\\
    \leq&\sum_{i=1}^{\frac{N}{2}}
    \lVert \begin{pmatrix}H_{2i-1}&H_{2i}\end{pmatrix}\rVert_2
    \Big\lVert \begin{pmatrix}\varepsilon_{2i-1}\\\varepsilon_{2i}\end{pmatrix}\Big\rVert_2\\
    \leq&\sqrt{2}\sum_{i=1}^{\frac{N}{2}}
    \lVert \begin{pmatrix}H_{2i-1}&H_{2i}\end{pmatrix}\rVert_2
    \lVert \varepsilon\rVert_{\infty},
\end{align*}
where we use that $\lVert a\rVert_2\leq\sqrt{2}\lVert a\rVert_{\infty}$ if $a\in\bbr^2$.
\end{proof}

\section{Circuits for the $3$-qubit QFT}\label{app:qft}
In the following, we present the transpiled quantum circuits for the $3$-qubit QFT for our considered elementary gate sets from Section~\ref{sec:robustness_elementary}.
All circuits were transpiled using the BQSKit transpiler~\cite{bqskit_software}.
Figure~\ref{fig:circuit_qft} shows the original, un-transpiled quantum circuit for the $3$-qubit QFT from~\cite{qiskit_software}.
Further, Figures~\ref{fig:circuit_ibm}--\ref{fig:circuit_honeywell} show the transpiled circuits with native gate sets A--E:
gate set A in Figure~\ref{fig:circuit_ibm}, gate set B in Figure~\ref{fig:circuit_rigetti},
gate set C in Figure~\ref{fig:circuit_ibm_old}, gate set D in Figure~\ref{fig:circuit_google},
and gate set E in Figure~\ref{fig:circuit_honeywell}.

\begin{figure}[hbt!]
    \begin{center}
    \includegraphics[width=\columnwidth]{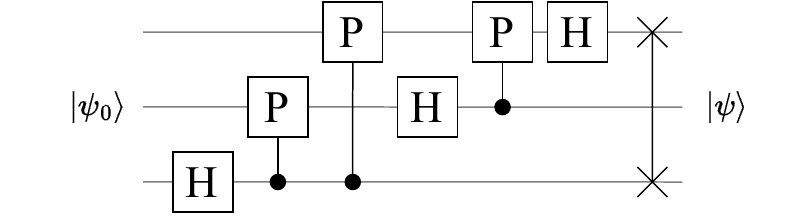}
    \end{center}
    \caption{
        The un-transpiled quantum circuit for the $3$-qubit QFT, which is taken from Qiskit~\cite{qiskit_software}.
        In total, the circuit consists of three Hadamard, three Controlled Phase and one SWAP gate.
        }
    \label{fig:circuit_qft}
\end{figure}

\begin{figure}[hbt!]
    \begin{center}
    \includegraphics[width=\columnwidth]{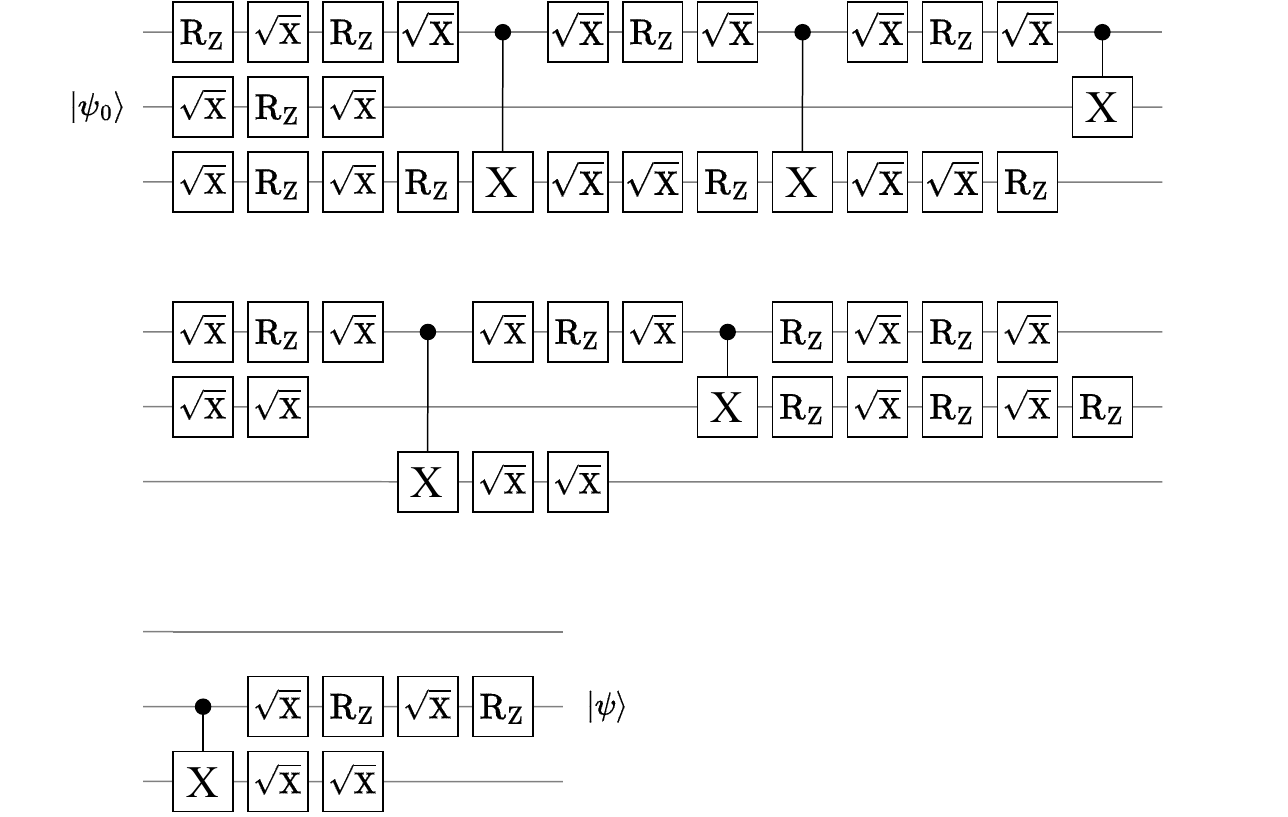}
    \end{center}
    \caption{
        The transpiled quantum circuit for the $3$-qubit QFT with gate set A.
        In total, the circuit consists of $18$ $R_{\rmZ}$, $30$ $\sqrt{X}$ and $6$ $CX$ gates and has a Lipschitz bound of $L_{QFT}^A = 117.95$.
        }
    \label{fig:circuit_ibm}
\end{figure}

\begin{figure}[hbt!]
    \begin{center}
    \includegraphics[width=\columnwidth]{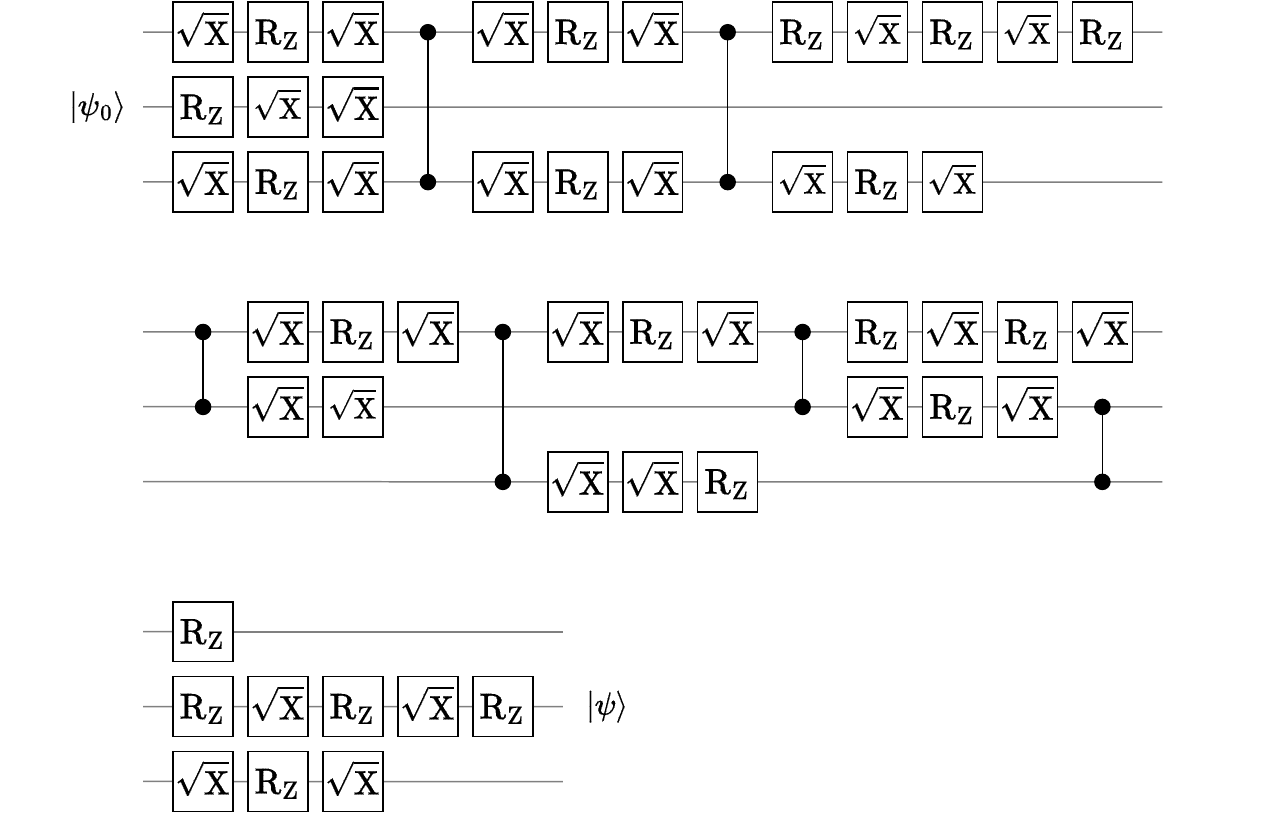}
    \end{center}
    \caption{
        The transpiled quantum circuit for the $3$-qubit QFT with gate set B.
        In total, the circuit consists of $20$ $R_{\rmZ}$, $30$ $\sqrt{X}$ and $6$ $CZ$ gates and has a Lipschitz bound of $L_{QFT}^B = 106.79$.
        }
    \label{fig:circuit_rigetti}
\end{figure}

\begin{figure}[hbt!]
    \begin{center}
    \includegraphics[width=\columnwidth]{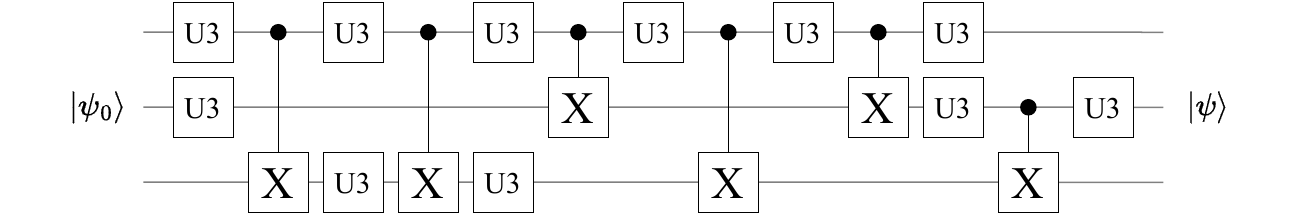}
    \end{center}
    \caption{
        The transpiled quantum circuit for the $3$-qubit QFT with gate set C.
        In total, the circuit consists of $11$ $U_3$ and $6$ $CX$ gates and has a Lipschitz bound of $L_{QFT}^C = 45.26$.
        }
    \label{fig:circuit_ibm_old}
\end{figure}

\begin{figure}[hbt!]
    \begin{center}
    \includegraphics[width=\columnwidth]{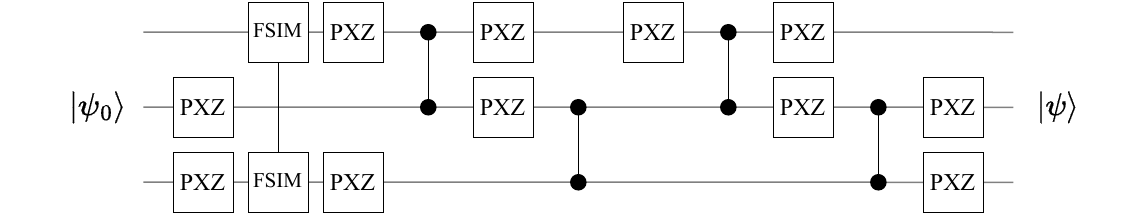}
    \end{center}
    \caption{
        The transpiled quantum circuit for the $3$-qubit QFT with gate set D.
        In total, the circuit consists of $11$ PhasedXZ and $4$ $CZ$ and one FSIM gate and has a Lipschitz bound of $L_{QFT}^D = 36.95$.
        }
    \label{fig:circuit_google}
\end{figure}

\begin{figure}[hbt!]
    \begin{center}
    \includegraphics[width=\columnwidth]{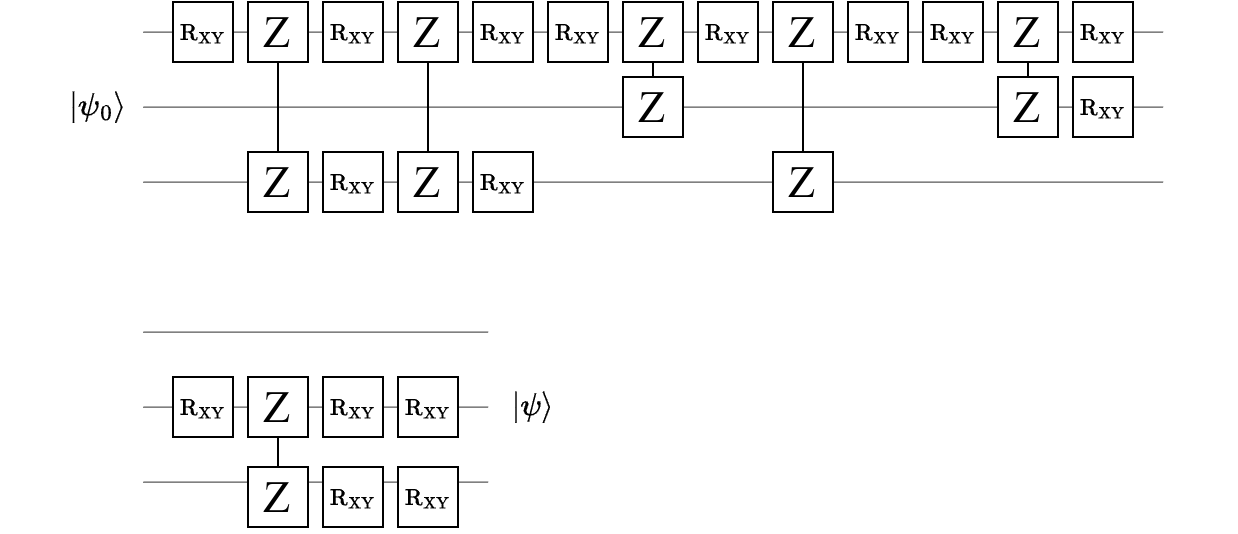}
    \end{center}
    \caption{
        The transpiled quantum circuit for the $3$-qubit QFT with gate set E.
        In total, the circuit consists of $16$ $R_{xy}$ and $6$ $U_{zz}$ gates and has a Lipschitz bound of $L_{QFT}^E = 29.42$.
        }
    \label{fig:circuit_honeywell}
\end{figure}

\section{Single qubit quantum state tomography}\label{app:qst}

In Section~\ref{sec:validation}, we compare an analytically known quantum state
$\ket{\hat{\psi}}$ with an unknown state $\ket{\psi(\varepsilon)}$, which is the output of the execution
of a quantum circuit on the real quantum computer \textit{ibm\_nairobi}.
Since the amplitudes of the quantum state cannot be measured directly, 
we need to use some technique to reconstruct the state based only on measurable expectation values.
This is the goal of a quantum state tomography (QST) procedure~\cite{schmied2016qst}.

In the considered setup, we are dealing with a single qubit state.
Therefore, we can write the corresponding density matrix as
\begin{align}\label{eq:qst_reconstruction}
    \rho = \frac{1}{2} \Bigl( I + r_x \cdot X + r_y \cdot Y + r_z \cdot Z \Bigr),
\end{align}
where the coefficients $r_x, r_y, r_z \in \mathbb{R}$ are the expectation values of the Pauli matrices, i.e.,
\begin{align}\label{eq:pauli_expectation_values}
    r_x &= \text{tr}(X \rho),
    \\
    r_y &= \text{tr}(Y \rho),
    \\
    r_z &= \text{tr}(Z \rho).
\end{align}
We can then obtain an estimate of these expectation values when preparing the quantum state $\rho$ several times
and measuring the projections onto the corresponding axes~\cite{schmied2016qst}.
Given the resulting density matrix $\rho$ and the analytical expression for the pure state $\ket{\psi}$,
we can calculate the fidelity according to
\begin{align}\label{eq:fidelity_density_matrix}
    \mathcal{F}(\ket{\psi}, \rho) 
    =
    \sqrt{| \braket{\psi | \rho | \psi} |}.
\end{align}

For our experiments, given a noise level $\bar{\varepsilon}$,
we draw a noise sample $\varepsilon \in [ - \bar{\varepsilon}, + \bar{\varepsilon} ]$,
prepare the state $\ket{\psi(\varepsilon)}$ $20,000$ times for each of the three Pauli expectation values,
reconstruct $\rho_{\varepsilon}$ according to~(\ref{eq:qst_reconstruction}),
and calculate the fidelity via~(\ref{eq:fidelity_density_matrix}).
We run this procedure for $16$ different noise levels, each featuring $80$ different noise samples, in order to collect the statistics for Figure~\ref{fig:experiment_real}.
The overall procedure is performed for each of the two quantum circuits $U_A$ and $U_B$.
Thus, we run $2 \cdot 16 \cdot 80 \cdot 3 \cdot 20,000 = 153,600,000$ quantum circuits on \textit{ibm\_nairobi} in total.

\bibliography{Literature}

\begin{thebibliography}{75}%
\makeatletter
\providecommand \@ifxundefined [1]{%
 \@ifx{#1\undefined}
}%
\providecommand \@ifnum [1]{%
 \ifnum #1\expandafter \@firstoftwo
 \else \expandafter \@secondoftwo
 \fi
}%
\providecommand \@ifx [1]{%
 \ifx #1\expandafter \@firstoftwo
 \else \expandafter \@secondoftwo
 \fi
}%
\providecommand \natexlab [1]{#1}%
\providecommand \enquote  [1]{``#1''}%
\providecommand \bibnamefont  [1]{#1}%
\providecommand \bibfnamefont [1]{#1}%
\providecommand \citenamefont [1]{#1}%
\providecommand \href@noop [0]{\@secondoftwo}%
\providecommand \href [0]{\begingroup \@sanitize@url \@href}%
\providecommand \@href[1]{\@@startlink{#1}\@@href}%
\providecommand \@@href[1]{\endgroup#1\@@endlink}%
\providecommand \@sanitize@url [0]{\catcode `\\12\catcode `\$12\catcode
  `\&12\catcode `\#12\catcode `\^12\catcode `\_12\catcode `\%12\relax}%
\providecommand \@@startlink[1]{}%
\providecommand \@@endlink[0]{}%
\providecommand \url  [0]{\begingroup\@sanitize@url \@url }%
\providecommand \@url [1]{\endgroup\@href {#1}{\urlprefix }}%
\providecommand \urlprefix  [0]{URL }%
\providecommand \Eprint [0]{\href }%
\providecommand \doibase [0]{https://doi.org/}%
\providecommand \selectlanguage [0]{\@gobble}%
\providecommand \bibinfo  [0]{\@secondoftwo}%
\providecommand \bibfield  [0]{\@secondoftwo}%
\providecommand \translation [1]{[#1]}%
\providecommand \BibitemOpen [0]{}%
\providecommand \bibitemStop [0]{}%
\providecommand \bibitemNoStop [0]{.\EOS\space}%
\providecommand \EOS [0]{\spacefactor3000\relax}%
\providecommand \BibitemShut  [1]{\csname bibitem#1\endcsname}%
\let\auto@bib@innerbib\@empty
\bibitem [{\citenamefont {Shor}(1999)}]{shor1999polynomial}%
  \BibitemOpen
  \bibfield  {author} {\bibinfo {author} {\bibfnamefont {P.}~\bibnamefont
  {Shor}},\ }\bibfield  {title} {\bibinfo {title} {Polynomial-time algorithms
  for prime factorization and discrete logarithms on a quantum computer},\
  }\href {https://doi.org/10.1137/S0036144598347011} {\bibfield  {journal}
  {\bibinfo  {journal} {SIAM Review}\ }\textbf {\bibinfo {volume} {41}},\
  \bibinfo {pages} {303} (\bibinfo {year} {1999})}\BibitemShut {NoStop}%
\bibitem [{\citenamefont {Grover}(1996)}]{grover1996fast}%
  \BibitemOpen
  \bibfield  {author} {\bibinfo {author} {\bibfnamefont {L.~K.}\ \bibnamefont
  {Grover}},\ }\bibfield  {title} {\bibinfo {title} {A fast quantum mechanical
  algorithm for database search},\ }in\ \href
  {https://doi.org/10.1145/237814.237866} {\emph {\bibinfo {booktitle} {Proc.
  28th ACM Symposium on the Theory of Computing}}}\ (\bibinfo {year} {1996})\
  pp.\ \bibinfo {pages} {212--219}\BibitemShut {NoStop}%
\bibitem [{\citenamefont {Preskill}(2018)}]{preskill2018quantum}%
  \BibitemOpen
  \bibfield  {author} {\bibinfo {author} {\bibfnamefont {J.}~\bibnamefont
  {Preskill}},\ }\bibfield  {title} {\bibinfo {title} {Quantum computing in the
  {NISQ} era and beyond},\ }\href {https://doi.org/10.22331/q-2018-08-06-79}
  {\bibfield  {journal} {\bibinfo  {journal} {Quantum}\ }\textbf {\bibinfo
  {volume} {2}},\ \bibinfo {pages} {79} (\bibinfo {year} {2018})}\BibitemShut
  {NoStop}%
\bibitem [{\citenamefont {Barnes}\ \emph {et~al.}(2017)\citenamefont {Barnes},
  \citenamefont {Trout},\ and\ \citenamefont {Lucarelli}}]{barnes2017quantum}%
  \BibitemOpen
  \bibfield  {author} {\bibinfo {author} {\bibfnamefont {J.~P.}\ \bibnamefont
  {Barnes}}, \bibinfo {author} {\bibfnamefont {C.~J.}\ \bibnamefont {Trout}},\
  and\ \bibinfo {author} {\bibfnamefont {D.}~\bibnamefont {Lucarelli}},\
  }\bibfield  {title} {\bibinfo {title} {Quantum error-correction failure
  distributions: comparison of coherent and stochastic error models},\ }\href
  {https://doi.org/10.1103/PhysRevA.95.062338} {\bibfield  {journal} {\bibinfo
  {journal} {Physical Review A}\ }\textbf {\bibinfo {volume} {95}},\ \bibinfo
  {pages} {062338} (\bibinfo {year} {2017})}\BibitemShut {NoStop}%
\bibitem [{\citenamefont {Trout}\ \emph {et~al.}(2018)\citenamefont {Trout},
  \citenamefont {Li}, \citenamefont {Guti{\'e}rrez}, \citenamefont {Wu},
  \citenamefont {Wang}, \citenamefont {Duan},\ and\ \citenamefont
  {Brown}}]{trout2018simulating}%
  \BibitemOpen
  \bibfield  {author} {\bibinfo {author} {\bibfnamefont {C.~J.}\ \bibnamefont
  {Trout}}, \bibinfo {author} {\bibfnamefont {M.}~\bibnamefont {Li}}, \bibinfo
  {author} {\bibfnamefont {M.}~\bibnamefont {Guti{\'e}rrez}}, \bibinfo {author}
  {\bibfnamefont {Y.}~\bibnamefont {Wu}}, \bibinfo {author} {\bibfnamefont
  {S.-T.}\ \bibnamefont {Wang}}, \bibinfo {author} {\bibfnamefont
  {L.}~\bibnamefont {Duan}},\ and\ \bibinfo {author} {\bibfnamefont {K.~R.}\
  \bibnamefont {Brown}},\ }\bibfield  {title} {\bibinfo {title} {Simulating the
  performance of a distance-3 surface code in a linear ion trap},\ }\href
  {https://doi.org/10.1088/1367-2630/aab341} {\bibfield  {journal} {\bibinfo
  {journal} {New Journal of Physics}\ }\textbf {\bibinfo {volume} {20}},\
  \bibinfo {pages} {043038} (\bibinfo {year} {2018})}\BibitemShut {NoStop}%
\bibitem [{\citenamefont {Arute}\ \emph {et~al.}(2019)\citenamefont {Arute},
  \citenamefont {Arya}, \citenamefont {Babbush},\ and\ \citenamefont
  {et~al.}}]{arute2019quantum}%
  \BibitemOpen
  \bibfield  {author} {\bibinfo {author} {\bibfnamefont {F.}~\bibnamefont
  {Arute}}, \bibinfo {author} {\bibfnamefont {K.}~\bibnamefont {Arya}},
  \bibinfo {author} {\bibfnamefont {R.}~\bibnamefont {Babbush}},\ and\ \bibinfo
  {author} {\bibnamefont {et~al.}},\ }\bibfield  {title} {\bibinfo {title}
  {Quantum supremacy using a programmable superconducting processor},\ }\href
  {https://doi.org/10.1038/s41586-019-1666-5} {\bibfield  {journal} {\bibinfo
  {journal} {Nature}\ }\textbf {\bibinfo {volume} {574}},\ \bibinfo {pages}
  {505} (\bibinfo {year} {2019})}\BibitemShut {NoStop}%
\bibitem [{\citenamefont {Levitt}(1986)}]{levitt1986composite}%
  \BibitemOpen
  \bibfield  {author} {\bibinfo {author} {\bibfnamefont {M.~H.}\ \bibnamefont
  {Levitt}},\ }\bibfield  {title} {\bibinfo {title} {Composite pulses},\ }\href
  {https://doi.org/10.1016/0079-6565(86)80005-X} {\bibfield  {journal}
  {\bibinfo  {journal} {Progress in NMR Spectroscopy}\ }\textbf {\bibinfo
  {volume} {18}},\ \bibinfo {pages} {61} (\bibinfo {year} {1986})}\BibitemShut
  {NoStop}%
\bibitem [{\citenamefont {Jones}(2003)}]{jones2003robust}%
  \BibitemOpen
  \bibfield  {author} {\bibinfo {author} {\bibfnamefont {J.~A.}\ \bibnamefont
  {Jones}},\ }\bibfield  {title} {\bibinfo {title} {Robust {Ising} gates for
  practical quantum computation},\ }\href
  {https://doi.org/10.1103/PhysRevA.67.012317} {\bibfield  {journal} {\bibinfo
  {journal} {Physical Review A}\ }\textbf {\bibinfo {volume} {67}},\ \bibinfo
  {pages} {012317} (\bibinfo {year} {2003})}\BibitemShut {NoStop}%
\bibitem [{\citenamefont {Brown}\ \emph {et~al.}(2004)\citenamefont {Brown},
  \citenamefont {Harrow},\ and\ \citenamefont {Chuang}}]{brown2004arbitrarily}%
  \BibitemOpen
  \bibfield  {author} {\bibinfo {author} {\bibfnamefont {K.~R.}\ \bibnamefont
  {Brown}}, \bibinfo {author} {\bibfnamefont {A.~W.}\ \bibnamefont {Harrow}},\
  and\ \bibinfo {author} {\bibfnamefont {I.~L.}\ \bibnamefont {Chuang}},\
  }\bibfield  {title} {\bibinfo {title} {Arbitrarily accurate composite pulse
  sequences},\ }\href@noop {} {\bibfield  {journal} {\bibinfo  {journal}
  {Physical Review A}\ }\textbf {\bibinfo {volume} {70}},\ \bibinfo {pages}
  {052318} (\bibinfo {year} {2004})}\BibitemShut {NoStop}%
\bibitem [{\citenamefont {Merrill}\ and\ \citenamefont
  {Brown}(2014)}]{merrill2014progress}%
  \BibitemOpen
  \bibfield  {author} {\bibinfo {author} {\bibfnamefont {J.~T.}\ \bibnamefont
  {Merrill}}\ and\ \bibinfo {author} {\bibfnamefont {K.~R.}\ \bibnamefont
  {Brown}},\ }\bibfield  {title} {\bibinfo {title} {Progress in compensating
  pulse sequences for quantum computation},\ }in\ \href
  {https://doi.org/10.1002/9781118742631.ch10} {\emph {\bibinfo {booktitle}
  {Advances in Chemical Physics}}},\ Vol.~\bibinfo {volume} {82}\ (\bibinfo
  {publisher} {John Wiley {\&} Sons, Hoboken, New Jersey},\ \bibinfo {year}
  {2014})\ p.\ \bibinfo {pages} {2014}\BibitemShut {NoStop}%
\bibitem [{\citenamefont {Khodjasteh}\ and\ \citenamefont
  {Viola}(2009)}]{khodjasteh2009dynamically}%
  \BibitemOpen
  \bibfield  {author} {\bibinfo {author} {\bibfnamefont {K.}~\bibnamefont
  {Khodjasteh}}\ and\ \bibinfo {author} {\bibfnamefont {L.}~\bibnamefont
  {Viola}},\ }\bibfield  {title} {\bibinfo {title} {Dynamically error-corrected
  gates for universal quantum computation},\ }\href
  {https://doi.org/10.1103/PhysRevLett.102.080501} {\bibfield  {journal}
  {\bibinfo  {journal} {Physical Review Letters}\ }\textbf {\bibinfo {volume}
  {102}},\ \bibinfo {pages} {080501} (\bibinfo {year} {2009})}\BibitemShut
  {NoStop}%
\bibitem [{\citenamefont {Bravyi}\ \emph {et~al.}(2018)\citenamefont {Bravyi},
  \citenamefont {Englbrecht}, \citenamefont {K{\"o}nig},\ and\ \citenamefont
  {Peard}}]{bravyi2018correcting}%
  \BibitemOpen
  \bibfield  {author} {\bibinfo {author} {\bibfnamefont {S.}~\bibnamefont
  {Bravyi}}, \bibinfo {author} {\bibfnamefont {M.}~\bibnamefont {Englbrecht}},
  \bibinfo {author} {\bibfnamefont {R.}~\bibnamefont {K{\"o}nig}},\ and\
  \bibinfo {author} {\bibfnamefont {N.}~\bibnamefont {Peard}},\ }\bibfield
  {title} {\bibinfo {title} {Correcting coherent errors with surface codes},\
  }\bibfield  {journal} {\bibinfo  {journal} {npj Quantum Information}\
  }\textbf {\bibinfo {volume} {4}},\ \href
  {https://doi.org/10.1038/s41534-018-0106-y} {10.1038/s41534-018-0106-y}
  (\bibinfo {year} {2018})\BibitemShut {NoStop}%
\bibitem [{\citenamefont {Debroy}\ \emph {et~al.}(2018)\citenamefont {Debroy},
  \citenamefont {Li}, \citenamefont {Newman},\ and\ \citenamefont
  {Brown}}]{debroy2018stabilizer}%
  \BibitemOpen
  \bibfield  {author} {\bibinfo {author} {\bibfnamefont {D.~M.}\ \bibnamefont
  {Debroy}}, \bibinfo {author} {\bibfnamefont {M.}~\bibnamefont {Li}}, \bibinfo
  {author} {\bibfnamefont {M.}~\bibnamefont {Newman}},\ and\ \bibinfo {author}
  {\bibfnamefont {K.~R.}\ \bibnamefont {Brown}},\ }\bibfield  {title} {\bibinfo
  {title} {Stabilizer slicing: coherent error cancellations in low-density
  parity-check stabilizer codes},\ }\href
  {https://doi.org/10.1103/PhysRevLett.121.250502} {\bibfield  {journal}
  {\bibinfo  {journal} {Physical Review Letters}\ }\textbf {\bibinfo {volume}
  {121}},\ \bibinfo {pages} {250502} (\bibinfo {year} {2018})}\BibitemShut
  {NoStop}%
\bibitem [{\citenamefont {Edmunds}\ \emph {et~al.}(2020)\citenamefont
  {Edmunds}, \citenamefont {Hempel}, \citenamefont {Harris}, \citenamefont
  {Frey}, \citenamefont {Stace},\ and\ \citenamefont
  {Biercuk}}]{edmunds2020dynamically}%
  \BibitemOpen
  \bibfield  {author} {\bibinfo {author} {\bibfnamefont {C.~L.}\ \bibnamefont
  {Edmunds}}, \bibinfo {author} {\bibfnamefont {C.}~\bibnamefont {Hempel}},
  \bibinfo {author} {\bibfnamefont {R.~J.}\ \bibnamefont {Harris}}, \bibinfo
  {author} {\bibfnamefont {V.}~\bibnamefont {Frey}}, \bibinfo {author}
  {\bibfnamefont {T.~M.}\ \bibnamefont {Stace}},\ and\ \bibinfo {author}
  {\bibfnamefont {M.~J.}\ \bibnamefont {Biercuk}},\ }\bibfield  {title}
  {\bibinfo {title} {Dynamically corrected gates suppressing spatiotemporal
  error correlations as measured by randomized benchmarking},\ }\href
  {https://doi.org/10.1103/PhysRevResearch.2.013156} {\bibfield  {journal}
  {\bibinfo  {journal} {Physical Review Research}\ }\textbf {\bibinfo {volume}
  {2}},\ \bibinfo {pages} {013156} (\bibinfo {year} {2020})}\BibitemShut
  {NoStop}%
\bibitem [{\citenamefont {Majumder}\ \emph {et~al.}(2020)\citenamefont
  {Majumder}, \citenamefont {{de Castro}},\ and\ \citenamefont
  {Brown}}]{majumder2020real}%
  \BibitemOpen
  \bibfield  {author} {\bibinfo {author} {\bibfnamefont {S.}~\bibnamefont
  {Majumder}}, \bibinfo {author} {\bibfnamefont {L.~A.}\ \bibnamefont {{de
  Castro}}},\ and\ \bibinfo {author} {\bibfnamefont {K.~R.}\ \bibnamefont
  {Brown}},\ }\bibfield  {title} {\bibinfo {title} {Real-time calibration with
  spectator qubits},\ }\href {https://doi.org/10.1038/s41534-020-0251-y}
  {\bibfield  {journal} {\bibinfo  {journal} {npj Quantum Information}\
  }\textbf {\bibinfo {volume} {6}},\ \bibinfo {pages} {19} (\bibinfo {year}
  {2020})}\BibitemShut {NoStop}%
\bibitem [{\citenamefont {Liu}(2022)}]{liu2022exact}%
  \BibitemOpen
  \bibfield  {author} {\bibinfo {author} {\bibfnamefont {C.}~\bibnamefont
  {Liu}},\ }\bibfield  {title} {\bibinfo {title} {Exact performance of the
  five-qubit code with coherent errors},\ }\bibfield  {journal} {\bibinfo
  {journal} {{arXiv:2203.01706}}\ }\href
  {https://doi.org/10.48550/arXiv.2203.01706} {10.48550/arXiv.2203.01706}
  (\bibinfo {year} {2022})\BibitemShut {NoStop}%
\bibitem [{\citenamefont {Wallman}\ and\ \citenamefont
  {Emerson}(2016)}]{wallman2016noise}%
  \BibitemOpen
  \bibfield  {author} {\bibinfo {author} {\bibfnamefont {J.~J.}\ \bibnamefont
  {Wallman}}\ and\ \bibinfo {author} {\bibfnamefont {J.}~\bibnamefont
  {Emerson}},\ }\bibfield  {title} {\bibinfo {title} {Noise tailoring for
  scalable quantum computation via randomized compiling},\ }\href
  {https://doi.org/10.1103/PhysRevA.94.052325} {\bibfield  {journal} {\bibinfo
  {journal} {Physical Review A}\ }\textbf {\bibinfo {volume} {94}},\ \bibinfo
  {pages} {052325} (\bibinfo {year} {2016})}\BibitemShut {NoStop}%
\bibitem [{\citenamefont {Nielsen}\ \emph {et~al.}(2021)\citenamefont
  {Nielsen}, \citenamefont {Gamble}, \citenamefont {Rudinger}, \citenamefont
  {Scholten}, \citenamefont {Young},\ and\ \citenamefont
  {Blume-Kohout}}]{nielsen2021gate}%
  \BibitemOpen
  \bibfield  {author} {\bibinfo {author} {\bibfnamefont {E.}~\bibnamefont
  {Nielsen}}, \bibinfo {author} {\bibfnamefont {J.~K.}\ \bibnamefont {Gamble}},
  \bibinfo {author} {\bibfnamefont {K.}~\bibnamefont {Rudinger}}, \bibinfo
  {author} {\bibfnamefont {T.}~\bibnamefont {Scholten}}, \bibinfo {author}
  {\bibfnamefont {K.}~\bibnamefont {Young}},\ and\ \bibinfo {author}
  {\bibfnamefont {R.}~\bibnamefont {Blume-Kohout}},\ }\bibfield  {title}
  {\bibinfo {title} {Gate set tomography},\ }\href@noop {} {\bibfield
  {journal} {\bibinfo  {journal} {Quantum}\ }\textbf {\bibinfo {volume} {5}},\
  \bibinfo {pages} {557} (\bibinfo {year} {2021})}\BibitemShut {NoStop}%
\bibitem [{\citenamefont {Zhang}\ \emph {et~al.}(2022)\citenamefont {Zhang},
  \citenamefont {Majumder}, \citenamefont {Leung}, \citenamefont {Crain},
  \citenamefont {Wang}, \citenamefont {Fang}, \citenamefont {Debroy},
  \citenamefont {Kim},\ and\ \citenamefont {Brown}}]{zhang2022hidden}%
  \BibitemOpen
  \bibfield  {author} {\bibinfo {author} {\bibfnamefont {B.}~\bibnamefont
  {Zhang}}, \bibinfo {author} {\bibfnamefont {S.}~\bibnamefont {Majumder}},
  \bibinfo {author} {\bibfnamefont {P.~H.}\ \bibnamefont {Leung}}, \bibinfo
  {author} {\bibfnamefont {S.}~\bibnamefont {Crain}}, \bibinfo {author}
  {\bibfnamefont {Y.}~\bibnamefont {Wang}}, \bibinfo {author} {\bibfnamefont
  {C.}~\bibnamefont {Fang}}, \bibinfo {author} {\bibfnamefont {D.~M.}\
  \bibnamefont {Debroy}}, \bibinfo {author} {\bibfnamefont {J.}~\bibnamefont
  {Kim}},\ and\ \bibinfo {author} {\bibfnamefont {K.~R.}\ \bibnamefont
  {Brown}},\ }\bibfield  {title} {\bibinfo {title} {Hidden inverses: coherent
  error cancellation at the circuit level},\ }\href
  {https://doi.org/10.1103/PhysRevApplied.17.034074} {\bibfield  {journal}
  {\bibinfo  {journal} {Physical Review Applied}\ }\textbf {\bibinfo {volume}
  {17}},\ \bibinfo {pages} {034074} (\bibinfo {year} {2022})}\BibitemShut
  {NoStop}%
\bibitem [{\citenamefont {Maksymov}\ \emph {et~al.}(2021)\citenamefont
  {Maksymov}, \citenamefont {Niroula},\ and\ \citenamefont
  {Nam}}]{maksymov2021optimal}%
  \BibitemOpen
  \bibfield  {author} {\bibinfo {author} {\bibfnamefont {A.}~\bibnamefont
  {Maksymov}}, \bibinfo {author} {\bibfnamefont {P.}~\bibnamefont {Niroula}},\
  and\ \bibinfo {author} {\bibfnamefont {Y.}~\bibnamefont {Nam}},\ }\bibfield
  {title} {\bibinfo {title} {Optimal calibration of gates in trapped-ion
  quantum computers},\ }\href@noop {} {\bibfield  {journal} {\bibinfo
  {journal} {Quantum Science and Technology}\ }\textbf {\bibinfo {volume}
  {6}},\ \bibinfo {pages} {034009} (\bibinfo {year} {2021})}\BibitemShut
  {NoStop}%
\bibitem [{\citenamefont {Maslov}\ \emph {et~al.}(2008)\citenamefont {Maslov},
  \citenamefont {Dueck}, \citenamefont {Miller},\ and\ \citenamefont
  {Negrevergne}}]{maslov2008quantum}%
  \BibitemOpen
  \bibfield  {author} {\bibinfo {author} {\bibfnamefont {D.}~\bibnamefont
  {Maslov}}, \bibinfo {author} {\bibfnamefont {G.~W.}\ \bibnamefont {Dueck}},
  \bibinfo {author} {\bibfnamefont {D.~M.}\ \bibnamefont {Miller}},\ and\
  \bibinfo {author} {\bibfnamefont {C.}~\bibnamefont {Negrevergne}},\
  }\bibfield  {title} {\bibinfo {title} {Quantum circuit simplification and
  level compaction},\ }\href {https://doi.org/10.1109/TCAD.2007.911334}
  {\bibfield  {journal} {\bibinfo  {journal} {IEEE Trans. Computer-Aided Design
  of Integrated Circuits and Systems}\ }\textbf {\bibinfo {volume} {27}},\
  \bibinfo {pages} {436} (\bibinfo {year} {2008})}\BibitemShut {NoStop}%
\bibitem [{\citenamefont {Arabzadeh}\ \emph {et~al.}(2010)\citenamefont
  {Arabzadeh}, \citenamefont {Saeedi},\ and\ \citenamefont
  {Zamani}}]{arabzadeh2010rule}%
  \BibitemOpen
  \bibfield  {author} {\bibinfo {author} {\bibfnamefont {M.}~\bibnamefont
  {Arabzadeh}}, \bibinfo {author} {\bibfnamefont {M.}~\bibnamefont {Saeedi}},\
  and\ \bibinfo {author} {\bibfnamefont {M.~S.}\ \bibnamefont {Zamani}},\
  }\bibfield  {title} {\bibinfo {title} {Rule-based optimization of reversible
  circuits},\ }in\ \href {https://doi.org/10.1109/ASPDAC.2010.5419684} {\emph
  {\bibinfo {booktitle} {Proc. 15th Asia and South Pacific Design Automation
  Conf. (ASP-DAC)}}}\ (\bibinfo {year} {2010})\ pp.\ \bibinfo {pages}
  {849--854}\BibitemShut {NoStop}%
\bibitem [{\citenamefont {Nam}\ \emph {et~al.}(2018)\citenamefont {Nam},
  \citenamefont {Ross}, \citenamefont {Su}, \citenamefont {Childs},\ and\
  \citenamefont {Maslov}}]{nam2018automated}%
  \BibitemOpen
  \bibfield  {author} {\bibinfo {author} {\bibfnamefont {Y.}~\bibnamefont
  {Nam}}, \bibinfo {author} {\bibfnamefont {N.~J.}\ \bibnamefont {Ross}},
  \bibinfo {author} {\bibfnamefont {Y.}~\bibnamefont {Su}}, \bibinfo {author}
  {\bibfnamefont {A.~M.}\ \bibnamefont {Childs}},\ and\ \bibinfo {author}
  {\bibfnamefont {D.}~\bibnamefont {Maslov}},\ }\bibfield  {title} {\bibinfo
  {title} {Automated optimization of large quantum circuits with continuous
  parameters},\ }\bibfield  {journal} {\bibinfo  {journal} {Quantum
  Information}\ }\textbf {\bibinfo {volume} {4}},\ \href
  {https://doi.org/10.1038/s41534-018-0072-4} {10.1038/s41534-018-0072-4}
  (\bibinfo {year} {2018})\BibitemShut {NoStop}%
\bibitem [{\citenamefont {Amy}\ \emph {et~al.}(2019)\citenamefont {Amy},
  \citenamefont {Azimzadeh},\ and\ \citenamefont {Mosca}}]{amy2019controlled}%
  \BibitemOpen
  \bibfield  {author} {\bibinfo {author} {\bibfnamefont {M.}~\bibnamefont
  {Amy}}, \bibinfo {author} {\bibfnamefont {P.}~\bibnamefont {Azimzadeh}},\
  and\ \bibinfo {author} {\bibfnamefont {M.}~\bibnamefont {Mosca}},\ }\bibfield
   {title} {\bibinfo {title} {On the controlled-{NOT} complexity of
  controlled-{NOT}-phase circuits},\ }\href
  {https://doi.org/10.1088/2058-9565/aad8ca} {\bibfield  {journal} {\bibinfo
  {journal} {Quantum Science and Technology}\ }\textbf {\bibinfo {volume}
  {4}},\ \bibinfo {pages} {015002} (\bibinfo {year} {2019})}\BibitemShut
  {NoStop}%
\bibitem [{\citenamefont {Lee}\ \emph {et~al.}(2019)\citenamefont {Lee},
  \citenamefont {Joo},\ and\ \citenamefont {Lee}}]{lee2019hybrid}%
  \BibitemOpen
  \bibfield  {author} {\bibinfo {author} {\bibfnamefont {Y.}~\bibnamefont
  {Lee}}, \bibinfo {author} {\bibfnamefont {J.}~\bibnamefont {Joo}},\ and\
  \bibinfo {author} {\bibfnamefont {S.}~\bibnamefont {Lee}},\ }\bibfield
  {title} {\bibinfo {title} {Hybrid quantum linear equation algorithm and its
  experimental test on {IBM Quantum Experience}},\ }\bibfield  {journal}
  {\bibinfo  {journal} {Scientific Reports}\ }\textbf {\bibinfo {volume} {9}},\
  \href {https://doi.org/10.1038/s41598-019-41324-9}
  {10.1038/s41598-019-41324-9} (\bibinfo {year} {2019})\BibitemShut {NoStop}%
\bibitem [{\citenamefont {Duncan}\ \emph {et~al.}(2020)\citenamefont {Duncan},
  \citenamefont {Kissinger}, \citenamefont {Perdrix},\ and\ \citenamefont {{van
  de Wetering}}}]{duncan2020graph}%
  \BibitemOpen
  \bibfield  {author} {\bibinfo {author} {\bibfnamefont {R.}~\bibnamefont
  {Duncan}}, \bibinfo {author} {\bibfnamefont {A.}~\bibnamefont {Kissinger}},
  \bibinfo {author} {\bibfnamefont {S.}~\bibnamefont {Perdrix}},\ and\ \bibinfo
  {author} {\bibfnamefont {J.}~\bibnamefont {{van de Wetering}}},\ }\bibfield
  {title} {\bibinfo {title} {Graph-theoretic simplification of quantum circuits
  with the {ZX}-calculus},\ }\href {https://doi.org/10.22331/q-2020-06-04-279}
  {\bibfield  {journal} {\bibinfo  {journal} {Quantum}\ }\textbf {\bibinfo
  {volume} {4}},\ \bibinfo {pages} {279} (\bibinfo {year} {2020})}\BibitemShut
  {NoStop}%
\bibitem [{\citenamefont {F{\"o}sel}\ \emph {et~al.}(2021)\citenamefont
  {F{\"o}sel}, \citenamefont {Niu}, \citenamefont {Marquardt},\ and\
  \citenamefont {Li}}]{foesel2021quantum}%
  \BibitemOpen
  \bibfield  {author} {\bibinfo {author} {\bibfnamefont {T.}~\bibnamefont
  {F{\"o}sel}}, \bibinfo {author} {\bibfnamefont {M.~Y.}\ \bibnamefont {Niu}},
  \bibinfo {author} {\bibfnamefont {F.}~\bibnamefont {Marquardt}},\ and\
  \bibinfo {author} {\bibfnamefont {L.}~\bibnamefont {Li}},\ }\bibfield
  {title} {\bibinfo {title} {Quantum circuit optimization with deep
  reinforcement learning},\ }\bibfield  {journal} {\bibinfo  {journal}
  {arXiv:2103.07585}\ }\href {https://doi.org/10.48550/arXiv.2103.07585}
  {10.48550/arXiv.2103.07585} (\bibinfo {year} {2021})\BibitemShut {NoStop}%
\bibitem [{\citenamefont {Nagarajan}\ \emph {et~al.}(2021)\citenamefont
  {Nagarajan}, \citenamefont {Lockwood},\ and\ \citenamefont
  {Coffrin}}]{nagarajan2021quantum}%
  \BibitemOpen
  \bibfield  {author} {\bibinfo {author} {\bibfnamefont {H.}~\bibnamefont
  {Nagarajan}}, \bibinfo {author} {\bibfnamefont {O.}~\bibnamefont
  {Lockwood}},\ and\ \bibinfo {author} {\bibfnamefont {C.}~\bibnamefont
  {Coffrin}},\ }\bibfield  {title} {\bibinfo {title} {{QuantumCircuitOpt}: an
  open-source framework for provably optimal quantum circuit design},\ }in\
  \href {https://doi.org/10.1109/QCS54837.2021.00010} {\emph {\bibinfo
  {booktitle} {Proc. 2nd IEEE/ACM Int. Workshop on Quantum Computing Software
  (QCS)}}}\ (\bibinfo {year} {2021})\ pp.\ \bibinfo {pages}
  {55--63}\BibitemShut {NoStop}%
\bibitem [{ibm({\natexlab{a}})}]{ibm_hardware}%
  \BibitemOpen
  \href@noop {} {\bibinfo {title} {{IBM Quantum.}
  \url{https://quantum-computing.ibm.com}, 2021}} ({\natexlab{a}})\BibitemShut
  {NoStop}%
\bibitem [{cod()}]{code_repository}%
  \BibitemOpen
  \href@noop {} {\bibinfo {title}
  {{\url{https://github.com/daniel-fink-de/robustness-of-quantum-algorithms}}}}\BibitemShut
  {NoStop}%
\bibitem [{\citenamefont {Skolik}\ \emph {et~al.}(2022)\citenamefont {Skolik},
  \citenamefont {Mangini}, \citenamefont {B{\"a}ck}, \citenamefont
  {Macchiavello},\ and\ \citenamefont {Dunjko}}]{skolik2022robustness}%
  \BibitemOpen
  \bibfield  {author} {\bibinfo {author} {\bibfnamefont {A.}~\bibnamefont
  {Skolik}}, \bibinfo {author} {\bibfnamefont {S.}~\bibnamefont {Mangini}},
  \bibinfo {author} {\bibfnamefont {T.}~\bibnamefont {B{\"a}ck}}, \bibinfo
  {author} {\bibfnamefont {C.}~\bibnamefont {Macchiavello}},\ and\ \bibinfo
  {author} {\bibfnamefont {V.}~\bibnamefont {Dunjko}},\ }\bibfield  {title}
  {\bibinfo {title} {Robustness of quantum reinforcement learning under
  hardware errors},\ }\bibfield  {journal} {\bibinfo  {journal}
  {{arXiv:2212.09431}}\ }\href {https://doi.org/10.48550/arXiv.2212.09431}
  {10.48550/arXiv.2212.09431} (\bibinfo {year} {2022})\BibitemShut {NoStop}%
\bibitem [{\citenamefont {Nielsen}\ and\ \citenamefont
  {Chuang}(2011)}]{nielsen2011quantum}%
  \BibitemOpen
  \bibfield  {author} {\bibinfo {author} {\bibfnamefont {M.~A.}\ \bibnamefont
  {Nielsen}}\ and\ \bibinfo {author} {\bibfnamefont {I.~L.}\ \bibnamefont
  {Chuang}},\ }\href {https://doi.org/10.1017/CBO9780511976667} {\emph
  {\bibinfo {title} {Quantum Computation and Quantum Information: 10th
  Anniversary Edition}}},\ \bibinfo {edition} {10th}\ ed.\ (\bibinfo
  {publisher} {Cambridge University Press, New York, NY, USA},\ \bibinfo {year}
  {2011})\BibitemShut {NoStop}%
\bibitem [{Note1()}]{Note1}%
  \BibitemOpen
  \bibinfo {note} {We use an $\infty $-norm for the noise rather than a
  $2$-norm in order to obtain tighter bounds in our main technical
  results.}\BibitemShut {Stop}%
\bibitem [{\citenamefont {Kitaev}(1997)}]{kitaev1997quantum}%
  \BibitemOpen
  \bibfield  {author} {\bibinfo {author} {\bibfnamefont {A.~Y.}\ \bibnamefont
  {Kitaev}},\ }\bibfield  {title} {\bibinfo {title} {Quantum computations:
  algorithms and error correction},\ }\href
  {https://doi.org/10.1070/RM1997v052n06ABEH002155} {\bibfield  {journal}
  {\bibinfo  {journal} {Russian Mathematical Surveys}\ }\textbf {\bibinfo
  {volume} {52}},\ \bibinfo {pages} {1191} (\bibinfo {year}
  {1997})}\BibitemShut {NoStop}%
\bibitem [{\citenamefont {Sweke}\ \emph {et~al.}(2020)\citenamefont {Sweke},
  \citenamefont {Wilde}, \citenamefont {Meyer}, \citenamefont {Schuld},
  \citenamefont {F{\"a}hrmann}, \citenamefont {Meynard-Piganeau},\ and\
  \citenamefont {Eisert}}]{sweke2020stochastic}%
  \BibitemOpen
  \bibfield  {author} {\bibinfo {author} {\bibfnamefont {R.}~\bibnamefont
  {Sweke}}, \bibinfo {author} {\bibfnamefont {F.}~\bibnamefont {Wilde}},
  \bibinfo {author} {\bibfnamefont {J.}~\bibnamefont {Meyer}}, \bibinfo
  {author} {\bibfnamefont {M.}~\bibnamefont {Schuld}}, \bibinfo {author}
  {\bibfnamefont {P.~K.}\ \bibnamefont {F{\"a}hrmann}}, \bibinfo {author}
  {\bibfnamefont {B.}~\bibnamefont {Meynard-Piganeau}},\ and\ \bibinfo {author}
  {\bibfnamefont {J.}~\bibnamefont {Eisert}},\ }\bibfield  {title} {\bibinfo
  {title} {Stochastic gradient descent for hybrid quantum-classical
  optimization},\ }\href {https://doi.org/10.22331/q-2020-08-31-314} {\bibfield
   {journal} {\bibinfo  {journal} {Quantum}\ }\textbf {\bibinfo {volume} {4}},\
  \bibinfo {pages} {314} (\bibinfo {year} {2020})}\BibitemShut {NoStop}%
\bibitem [{\citenamefont {Lidar}\ \emph {et~al.}(2008)\citenamefont {Lidar},
  \citenamefont {Zanardi},\ and\ \citenamefont
  {Khodjasteh}}]{lidar2008distance}%
  \BibitemOpen
  \bibfield  {author} {\bibinfo {author} {\bibfnamefont {D.~A.}\ \bibnamefont
  {Lidar}}, \bibinfo {author} {\bibfnamefont {P.}~\bibnamefont {Zanardi}},\
  and\ \bibinfo {author} {\bibfnamefont {K.}~\bibnamefont {Khodjasteh}},\
  }\bibfield  {title} {\bibinfo {title} {Distance bounds on quantum dynamics},\
  }\href {https://doi.org/10.1103/PhysRevA.78.012308} {\bibfield  {journal}
  {\bibinfo  {journal} {Physical Review A}\ }\textbf {\bibinfo {volume} {78}},\
  \bibinfo {pages} {012308} (\bibinfo {year} {2008})}\BibitemShut {NoStop}%
\bibitem [{\citenamefont {Aharonov}\ and\ \citenamefont
  {{Ben-Or}}(1997)}]{aharonov1997fault}%
  \BibitemOpen
  \bibfield  {author} {\bibinfo {author} {\bibfnamefont {D.}~\bibnamefont
  {Aharonov}}\ and\ \bibinfo {author} {\bibfnamefont {M.}~\bibnamefont
  {{Ben-Or}}},\ }\bibfield  {title} {\bibinfo {title} {Faul-tolerant quantum
  computation with constant error},\ }in\ \href
  {https://doi.org/10.1145/258533.258579} {\emph {\bibinfo {booktitle} {Proc.
  29th ACM Symposium on the Theory of Computing}}}\ (\bibinfo {year} {1997})\
  pp.\ \bibinfo {pages} {176--188}\BibitemShut {NoStop}%
\bibitem [{\citenamefont {Aharonov}\ and\ \citenamefont
  {{Ben-Or}}(2008)}]{aharonov2008fault}%
  \BibitemOpen
  \bibfield  {author} {\bibinfo {author} {\bibfnamefont {D.}~\bibnamefont
  {Aharonov}}\ and\ \bibinfo {author} {\bibfnamefont {M.}~\bibnamefont
  {{Ben-Or}}},\ }\bibfield  {title} {\bibinfo {title} {Fault-tolerant quantum
  computation with constant error rate},\ }\href
  {https://doi.org/10.1137/S0097539799359385} {\bibfield  {journal} {\bibinfo
  {journal} {SIAM J. Comput.}\ }\textbf {\bibinfo {volume} {38}},\ \bibinfo
  {pages} {1207} (\bibinfo {year} {2008})}\BibitemShut {NoStop}%
\bibitem [{\citenamefont {Blume-Kohout}\ \emph {et~al.}(2017)\citenamefont
  {Blume-Kohout}, \citenamefont {Gamble}, \citenamefont {Nielsen},
  \citenamefont {Rudinger}, \citenamefont {Mizrahi}, \citenamefont {Fortier},\
  and\ \citenamefont {Maunz}}]{blume2017demonstration}%
  \BibitemOpen
  \bibfield  {author} {\bibinfo {author} {\bibfnamefont {R.}~\bibnamefont
  {Blume-Kohout}}, \bibinfo {author} {\bibfnamefont {J.~K.}\ \bibnamefont
  {Gamble}}, \bibinfo {author} {\bibfnamefont {E.}~\bibnamefont {Nielsen}},
  \bibinfo {author} {\bibfnamefont {K.}~\bibnamefont {Rudinger}}, \bibinfo
  {author} {\bibfnamefont {J.}~\bibnamefont {Mizrahi}}, \bibinfo {author}
  {\bibfnamefont {K.}~\bibnamefont {Fortier}},\ and\ \bibinfo {author}
  {\bibfnamefont {P.}~\bibnamefont {Maunz}},\ }\bibfield  {title} {\bibinfo
  {title} {Demonstration of qubit operations below a rigorous fault tolerance
  threshold with gate set tomography},\ }\href
  {https://doi.org/10.1038/ncomms14485} {\bibfield  {journal} {\bibinfo
  {journal} {Nature Communications}\ }\textbf {\bibinfo {volume} {8}},\
  \bibinfo {pages} {14485} (\bibinfo {year} {2017})}\BibitemShut {NoStop}%
\bibitem [{\citenamefont {Nielsen}(2002)}]{nielsen2002simple}%
  \BibitemOpen
  \bibfield  {author} {\bibinfo {author} {\bibfnamefont {M.~A.}\ \bibnamefont
  {Nielsen}},\ }\bibfield  {title} {\bibinfo {title} {A simple formula for the
  average gate fidelity of a quantum dynamical operation},\ }\href
  {https://doi.org/10.1016/S0375-9601(02)01272-0} {\bibfield  {journal}
  {\bibinfo  {journal} {Physics Letters A}\ }\textbf {\bibinfo {volume}
  {303}},\ \bibinfo {pages} {249} (\bibinfo {year} {2002})}\BibitemShut
  {NoStop}%
\bibitem [{\citenamefont {Thomas}\ \emph {et~al.}(2011)\citenamefont {Thomas},
  \citenamefont {Lababidi},\ and\ \citenamefont {Tian}}]{thomas2011robustness}%
  \BibitemOpen
  \bibfield  {author} {\bibinfo {author} {\bibfnamefont {J.~T.}\ \bibnamefont
  {Thomas}}, \bibinfo {author} {\bibfnamefont {M.}~\bibnamefont {Lababidi}},\
  and\ \bibinfo {author} {\bibfnamefont {M.}~\bibnamefont {Tian}},\ }\bibfield
  {title} {\bibinfo {title} {Robustness of single-qubit geometric gate against
  systematic error},\ }\href {https://doi.org/10.1103/PhysRevA.84.042335}
  {\bibfield  {journal} {\bibinfo  {journal} {Physical Review A}\ }\textbf
  {\bibinfo {volume} {84}},\ \bibinfo {pages} {042335} (\bibinfo {year}
  {2011})}\BibitemShut {NoStop}%
\bibitem [{\citenamefont {Beigi}\ and\ \citenamefont
  {K{\"o}nig}(2011)}]{beigi2011simplified}%
  \BibitemOpen
  \bibfield  {author} {\bibinfo {author} {\bibfnamefont {S.}~\bibnamefont
  {Beigi}}\ and\ \bibinfo {author} {\bibfnamefont {R.}~\bibnamefont
  {K{\"o}nig}},\ }\bibfield  {title} {\bibinfo {title} {Simplified
  instantaneous non-local quantum computation with applications to
  position-based cryptography},\ }\href
  {https://doi.org/10.1088/1367-2630/13/9/093036} {\bibfield  {journal}
  {\bibinfo  {journal} {New Journal of Physics}\ }\textbf {\bibinfo {volume}
  {13}},\ \bibinfo {pages} {093036} (\bibinfo {year} {2011})}\BibitemShut
  {NoStop}%
\bibitem [{\citenamefont {Wallman}\ and\ \citenamefont
  {Flammia}(2014)}]{wallman2014randomized}%
  \BibitemOpen
  \bibfield  {author} {\bibinfo {author} {\bibfnamefont {J.~J.}\ \bibnamefont
  {Wallman}}\ and\ \bibinfo {author} {\bibfnamefont {S.~T.}\ \bibnamefont
  {Flammia}},\ }\bibfield  {title} {\bibinfo {title} {Randomized benchmarking
  with confidence},\ }\href {https://doi.org/10.1088/1367-2630/16/10/103032}
  {\bibfield  {journal} {\bibinfo  {journal} {New Journals of Physics}\
  }\textbf {\bibinfo {volume} {16}},\ \bibinfo {pages} {103032} (\bibinfo
  {year} {2014})}\BibitemShut {NoStop}%
\bibitem [{\citenamefont {Wallman}(2015)}]{wallman2015error}%
  \BibitemOpen
  \bibfield  {author} {\bibinfo {author} {\bibfnamefont {J.}~\bibnamefont
  {Wallman}},\ }\bibfield  {title} {\bibinfo {title} {Error rates in quantum
  circuits},\ }\bibfield  {journal} {\bibinfo  {journal} {{arXiv:1511.00727}}\
  }\href {https://doi.org/10.48550/arXiv.1511.00727}
  {10.48550/arXiv.1511.00727} (\bibinfo {year} {2015})\BibitemShut {NoStop}%
\bibitem [{\citenamefont {Wallman}\ \emph {et~al.}(2015)\citenamefont
  {Wallman}, \citenamefont {Granade}, \citenamefont {Harper},\ and\
  \citenamefont {Flammia}}]{wallman2015estimating}%
  \BibitemOpen
  \bibfield  {author} {\bibinfo {author} {\bibfnamefont {J.}~\bibnamefont
  {Wallman}}, \bibinfo {author} {\bibfnamefont {C.}~\bibnamefont {Granade}},
  \bibinfo {author} {\bibfnamefont {R.}~\bibnamefont {Harper}},\ and\ \bibinfo
  {author} {\bibfnamefont {S.~T.}\ \bibnamefont {Flammia}},\ }\bibfield
  {title} {\bibinfo {title} {Estimating the coherence of noise},\ }\href
  {https://doi.org/10.1088/1367-2630/17/11/113020} {\bibfield  {journal}
  {\bibinfo  {journal} {New Journal of Physics}\ }\textbf {\bibinfo {volume}
  {17}},\ \bibinfo {pages} {113020} (\bibinfo {year} {2015})}\BibitemShut
  {NoStop}%
\bibitem [{\citenamefont {Kueng}\ \emph {et~al.}(2016)\citenamefont {Kueng},
  \citenamefont {Long}, \citenamefont {Doherty},\ and\ \citenamefont
  {Flammia}}]{kueng2016comparing}%
  \BibitemOpen
  \bibfield  {author} {\bibinfo {author} {\bibfnamefont {R.}~\bibnamefont
  {Kueng}}, \bibinfo {author} {\bibfnamefont {D.~M.}\ \bibnamefont {Long}},
  \bibinfo {author} {\bibfnamefont {A.~C.}\ \bibnamefont {Doherty}},\ and\
  \bibinfo {author} {\bibfnamefont {S.~T.}\ \bibnamefont {Flammia}},\
  }\bibfield  {title} {\bibinfo {title} {Comparing experiments to the
  fault-tolerance threshold},\ }\href
  {https://doi.org/10.1103/PhysRevLett.117.170502} {\bibfield  {journal}
  {\bibinfo  {journal} {Physical Review Letters}\ }\textbf {\bibinfo {volume}
  {117}},\ \bibinfo {pages} {170502} (\bibinfo {year} {2016})}\BibitemShut
  {NoStop}%
\bibitem [{\citenamefont {Sanders}\ \emph {et~al.}(2016)\citenamefont
  {Sanders}, \citenamefont {Wallman},\ and\ \citenamefont
  {Sanders}}]{sanders2016bounding}%
  \BibitemOpen
  \bibfield  {author} {\bibinfo {author} {\bibfnamefont {Y.~R.}\ \bibnamefont
  {Sanders}}, \bibinfo {author} {\bibfnamefont {J.~J.}\ \bibnamefont
  {Wallman}},\ and\ \bibinfo {author} {\bibfnamefont {B.~C.}\ \bibnamefont
  {Sanders}},\ }\bibfield  {title} {\bibinfo {title} {Bounding quantum gate
  error rate based on reported average fidelity},\ }\href
  {https://doi.org/10.1088/1367-2630/18/1/012002} {\bibfield  {journal}
  {\bibinfo  {journal} {New Journal of Physics}\ }\textbf {\bibinfo {volume}
  {18}},\ \bibinfo {pages} {012002} (\bibinfo {year} {2016})}\BibitemShut
  {NoStop}%
\bibitem [{\citenamefont {Shor}(1995)}]{shor1995scheme}%
  \BibitemOpen
  \bibfield  {author} {\bibinfo {author} {\bibfnamefont {P.~W.}\ \bibnamefont
  {Shor}},\ }\bibfield  {title} {\bibinfo {title} {Scheme for reducing
  decoherence in quantum computer memory},\ }\href
  {https://doi.org/10.1103/PhysRevA.52.R2493} {\bibfield  {journal} {\bibinfo
  {journal} {Physical Review A}\ }\textbf {\bibinfo {volume} {52}},\ \bibinfo
  {pages} {R2493} (\bibinfo {year} {1995})}\BibitemShut {NoStop}%
\bibitem [{\citenamefont {Calderbank}\ and\ \citenamefont
  {Shor}(1996)}]{calderbank1996good}%
  \BibitemOpen
  \bibfield  {author} {\bibinfo {author} {\bibfnamefont {A.~R.}\ \bibnamefont
  {Calderbank}}\ and\ \bibinfo {author} {\bibfnamefont {P.~W.}\ \bibnamefont
  {Shor}},\ }\bibfield  {title} {\bibinfo {title} {Good quantum
  error-correcting codes exist},\ }\href
  {https://doi.org/10.1103/PhysRevA.54.1098} {\bibfield  {journal} {\bibinfo
  {journal} {Physical Review A}\ }\textbf {\bibinfo {volume} {54}},\ \bibinfo
  {pages} {1098} (\bibinfo {year} {1996})}\BibitemShut {NoStop}%
\bibitem [{\citenamefont {Steane}(1996)}]{steane1996error}%
  \BibitemOpen
  \bibfield  {author} {\bibinfo {author} {\bibfnamefont {A.~M.}\ \bibnamefont
  {Steane}},\ }\bibfield  {title} {\bibinfo {title} {Error correcting codes in
  quantum theory},\ }\href {https://doi.org/10.1103/PhysRevLett.77.793}
  {\bibfield  {journal} {\bibinfo  {journal} {Physical Review Letters}\
  }\textbf {\bibinfo {volume} {77}},\ \bibinfo {pages} {793} (\bibinfo {year}
  {1996})}\BibitemShut {NoStop}%
\bibitem [{\citenamefont {Cai}\ \emph {et~al.}(2022)\citenamefont {Cai},
  \citenamefont {Babbush}, \citenamefont {Benjamin}, \citenamefont {Endo},
  \citenamefont {Huggins}, \citenamefont {Li}, \citenamefont {{McClean}},\ and\
  \citenamefont {{O'Brien}}}]{cai2022quantum}%
  \BibitemOpen
  \bibfield  {author} {\bibinfo {author} {\bibfnamefont {Z.}~\bibnamefont
  {Cai}}, \bibinfo {author} {\bibfnamefont {R.}~\bibnamefont {Babbush}},
  \bibinfo {author} {\bibfnamefont {S.~C.}\ \bibnamefont {Benjamin}}, \bibinfo
  {author} {\bibfnamefont {S.}~\bibnamefont {Endo}}, \bibinfo {author}
  {\bibfnamefont {W.~J.}\ \bibnamefont {Huggins}}, \bibinfo {author}
  {\bibfnamefont {Y.}~\bibnamefont {Li}}, \bibinfo {author} {\bibfnamefont
  {J.~R.}\ \bibnamefont {{McClean}}},\ and\ \bibinfo {author} {\bibfnamefont
  {T.~E.}\ \bibnamefont {{O'Brien}}},\ }\bibfield  {title} {\bibinfo {title}
  {Quantum error mitigation},\ }\bibfield  {journal} {\bibinfo  {journal}
  {{arXiv:2210.00921}}\ }\href {https://doi.org/10.48550/arXiv.2210.00921}
  {10.48550/arXiv.2210.00921} (\bibinfo {year} {2022})\BibitemShut {NoStop}%
\bibitem [{\citenamefont {Takagi}\ \emph {et~al.}(2022)\citenamefont {Takagi},
  \citenamefont {Endo}, \citenamefont {Minagawa},\ and\ \citenamefont
  {Gu}}]{takagi2022fundamental}%
  \BibitemOpen
  \bibfield  {author} {\bibinfo {author} {\bibfnamefont {R.}~\bibnamefont
  {Takagi}}, \bibinfo {author} {\bibfnamefont {S.}~\bibnamefont {Endo}},
  \bibinfo {author} {\bibfnamefont {S.}~\bibnamefont {Minagawa}},\ and\
  \bibinfo {author} {\bibfnamefont {M.}~\bibnamefont {Gu}},\ }\bibfield
  {title} {\bibinfo {title} {Fundamental limits of quantum error mitigation},\
  }\bibfield  {journal} {\bibinfo  {journal} {Quantum Information}\ }\textbf
  {\bibinfo {volume} {8}},\ \href {https://doi.org/10.1038/s41534-022-00618-z}
  {10.1038/s41534-022-00618-z} (\bibinfo {year} {2022})\BibitemShut {NoStop}%
\bibitem [{\citenamefont {Quek}\ \emph {et~al.}(2022)\citenamefont {Quek},
  \citenamefont {Franca}, \citenamefont {Khatri}, \citenamefont {Meyer},\ and\
  \citenamefont {Eisert}}]{quek2022exponentially}%
  \BibitemOpen
  \bibfield  {author} {\bibinfo {author} {\bibfnamefont {Y.}~\bibnamefont
  {Quek}}, \bibinfo {author} {\bibfnamefont {D.~S.}\ \bibnamefont {Franca}},
  \bibinfo {author} {\bibfnamefont {S.}~\bibnamefont {Khatri}}, \bibinfo
  {author} {\bibfnamefont {J.~J.}\ \bibnamefont {Meyer}},\ and\ \bibinfo
  {author} {\bibfnamefont {J.}~\bibnamefont {Eisert}},\ }\bibfield  {title}
  {\bibinfo {title} {Exponentially tighter bounds on limitations of quantum
  error mitigation},\ }\bibfield  {journal} {\bibinfo  {journal}
  {{arXiv:2210.11505}}\ }\href {https://doi.org/10.48550/arXiv.2210.11505}
  {10.48550/arXiv.2210.11505} (\bibinfo {year} {2022})\BibitemShut {NoStop}%
\bibitem [{ibm({\natexlab{b}})}]{ibm_gateset}%
  \BibitemOpen
  \href@noop {} {\bibinfo {title}
  {{\url{https://quantum-computing.ibm.com/services/resources}}}}
  ({\natexlab{b}})\BibitemShut {NoStop}%
\bibitem [{rig()}]{rigetti_gateset}%
  \BibitemOpen
  \href@noop {} {\bibinfo {title}
  {{\url{https://pyquil-docs.rigetti.com/en/v2.7.0/apidocs/gates.html}}}}\BibitemShut
  {NoStop}%
\bibitem [{ibm({\natexlab{c}})}]{ibm_gateset_old}%
  \BibitemOpen
  \href@noop {} {\bibinfo {title}
  {{\url{https://github.com/Qiskit/ibmq-device-information/blob/master/backends/melbourne/V1/README.md}}}}
  ({\natexlab{c}})\BibitemShut {NoStop}%
\bibitem [{goo()}]{google_gateset}%
  \BibitemOpen
  \href@noop {} {\bibinfo {title}
  {{\url{https://quantumai.google/cirq/google/devices#sycamore}}}}\BibitemShut
  {NoStop}%
\bibitem [{\citenamefont {Pino}\ \emph {et~al.}(2021)\citenamefont {Pino},
  \citenamefont {Dreiling}, \citenamefont {Figgatt}, \citenamefont {Gaebler},
  \citenamefont {Moses}, \citenamefont {Allman}, \citenamefont {Baldwin},
  \citenamefont {Foss-Feig}, \citenamefont {Hayes}, \citenamefont {Mayer},
  \citenamefont {Ryan-Anderson},\ and\ \citenamefont
  {Neyenhuis}}]{honeywell_gateset}%
  \BibitemOpen
  \bibfield  {author} {\bibinfo {author} {\bibfnamefont {J.~M.}\ \bibnamefont
  {Pino}}, \bibinfo {author} {\bibfnamefont {J.~M.}\ \bibnamefont {Dreiling}},
  \bibinfo {author} {\bibfnamefont {C.}~\bibnamefont {Figgatt}}, \bibinfo
  {author} {\bibfnamefont {J.~P.}\ \bibnamefont {Gaebler}}, \bibinfo {author}
  {\bibfnamefont {S.~A.}\ \bibnamefont {Moses}}, \bibinfo {author}
  {\bibfnamefont {M.~S.}\ \bibnamefont {Allman}}, \bibinfo {author}
  {\bibfnamefont {C.~H.}\ \bibnamefont {Baldwin}}, \bibinfo {author}
  {\bibfnamefont {M.}~\bibnamefont {Foss-Feig}}, \bibinfo {author}
  {\bibfnamefont {D.}~\bibnamefont {Hayes}}, \bibinfo {author} {\bibfnamefont
  {K.}~\bibnamefont {Mayer}}, \bibinfo {author} {\bibfnamefont
  {C.}~\bibnamefont {Ryan-Anderson}},\ and\ \bibinfo {author} {\bibfnamefont
  {B.}~\bibnamefont {Neyenhuis}},\ }\bibfield  {title} {\bibinfo {title}
  {Demonstration of the trapped-ion quantum ccd computer architecture},\ }\href
  {https://doi.org/10.1038/s41586-021-03318-4} {\bibfield  {journal} {\bibinfo
  {journal} {Nature}\ }\textbf {\bibinfo {volume} {592}},\ \bibinfo {pages}
  {209} (\bibinfo {year} {2021})}\BibitemShut {NoStop}%
\bibitem [{\citenamefont {Younis}\ \emph {et~al.}(2021)\citenamefont {Younis},
  \citenamefont {Iancu}, \citenamefont {Lavrijsen}, \citenamefont {Davis},
  \citenamefont {Smith},\ and\ \citenamefont {USDOE}}]{bqskit_software}%
  \BibitemOpen
  \bibfield  {author} {\bibinfo {author} {\bibfnamefont {E.}~\bibnamefont
  {Younis}}, \bibinfo {author} {\bibfnamefont {C.~C.}\ \bibnamefont {Iancu}},
  \bibinfo {author} {\bibfnamefont {W.}~\bibnamefont {Lavrijsen}}, \bibinfo
  {author} {\bibfnamefont {M.}~\bibnamefont {Davis}}, \bibinfo {author}
  {\bibfnamefont {E.}~\bibnamefont {Smith}},\ and\ \bibinfo {author}
  {\bibnamefont {USDOE}},\ }\href {https://doi.org/10.11578/dc.20210603.2}
  {\bibinfo {title} {Berkeley quantum synthesis toolkit (bqskit) v1}} (\bibinfo
  {year} {2021})\BibitemShut {NoStop}%
\bibitem [{Note2()}]{Note2}%
  \BibitemOpen
  \bibinfo {note} {It is easy to show that the additional gates $X$ and
  $\protect \sqrt {X}$ do not affect the Lipschitz bounds w.r.t.\ the coherent
  control errors in $R_\protect \mathrm {z}$.}\BibitemShut {Stop}%
\bibitem [{\citenamefont {Gokhale}\ \emph {et~al.}(2020)\citenamefont
  {Gokhale}, \citenamefont {Javadi-Abhari}, \citenamefont {Earnest},
  \citenamefont {Shi},\ and\ \citenamefont {Chong}}]{gokhale2020optimized}%
  \BibitemOpen
  \bibfield  {author} {\bibinfo {author} {\bibfnamefont {P.}~\bibnamefont
  {Gokhale}}, \bibinfo {author} {\bibfnamefont {A.}~\bibnamefont
  {Javadi-Abhari}}, \bibinfo {author} {\bibfnamefont {N.}~\bibnamefont
  {Earnest}}, \bibinfo {author} {\bibfnamefont {Y.}~\bibnamefont {Shi}},\ and\
  \bibinfo {author} {\bibfnamefont {F.~T.}\ \bibnamefont {Chong}},\ }\bibfield
  {title} {\bibinfo {title} {Optimized quantum compilation for near-term
  algorithms with {OpenPulse}},\ }\href@noop {} {\bibfield  {journal} {\bibinfo
   {journal} {arXiv:2004.11205}\ } (\bibinfo {year} {2020})}\BibitemShut
  {NoStop}%
\bibitem [{\citenamefont {Cerezo}\ \emph {et~al.}(2021)\citenamefont {Cerezo},
  \citenamefont {Arrasmith}, \citenamefont {Babbush}, \citenamefont {Benjamin},
  \citenamefont {Endo}, \citenamefont {Fujii}, \citenamefont {McClean},
  \citenamefont {Mitarai}, \citenamefont {Yuan}, \citenamefont {Cincio},\ and\
  \citenamefont {Coles}}]{cerezo2021variational}%
  \BibitemOpen
  \bibfield  {author} {\bibinfo {author} {\bibfnamefont {M.}~\bibnamefont
  {Cerezo}}, \bibinfo {author} {\bibfnamefont {A.}~\bibnamefont {Arrasmith}},
  \bibinfo {author} {\bibfnamefont {R.}~\bibnamefont {Babbush}}, \bibinfo
  {author} {\bibfnamefont {S.~C.}\ \bibnamefont {Benjamin}}, \bibinfo {author}
  {\bibfnamefont {S.}~\bibnamefont {Endo}}, \bibinfo {author} {\bibfnamefont
  {K.}~\bibnamefont {Fujii}}, \bibinfo {author} {\bibfnamefont {J.~R.}\
  \bibnamefont {McClean}}, \bibinfo {author} {\bibfnamefont {K.}~\bibnamefont
  {Mitarai}}, \bibinfo {author} {\bibfnamefont {X.}~\bibnamefont {Yuan}},
  \bibinfo {author} {\bibfnamefont {L.}~\bibnamefont {Cincio}},\ and\ \bibinfo
  {author} {\bibfnamefont {P.~J.}\ \bibnamefont {Coles}},\ }\bibfield  {title}
  {\bibinfo {title} {Variational quantum algorithms},\ }\href
  {https://doi.org/10.1038/s42254-021-00348-9} {\bibfield  {journal} {\bibinfo
  {journal} {Nature Reviews Physics}\ }\textbf {\bibinfo {volume} {3}},\
  \bibinfo {pages} {625} (\bibinfo {year} {2021})}\BibitemShut {NoStop}%
\bibitem [{\citenamefont {Farhi}\ \emph {et~al.}(2014)\citenamefont {Farhi},
  \citenamefont {Goldstone},\ and\ \citenamefont {Gutmann}}]{farhi2014quantum}%
  \BibitemOpen
  \bibfield  {author} {\bibinfo {author} {\bibfnamefont {E.}~\bibnamefont
  {Farhi}}, \bibinfo {author} {\bibfnamefont {J.}~\bibnamefont {Goldstone}},\
  and\ \bibinfo {author} {\bibfnamefont {S.}~\bibnamefont {Gutmann}},\
  }\bibfield  {title} {\bibinfo {title} {A quantum approximate optimization
  algorithm},\ }\bibfield  {journal} {\bibinfo  {journal} {{arXiv:1411.4028}}\
  }\href {https://doi.org/10.48550/arXiv.1411.4028} {10.48550/arXiv.1411.4028}
  (\bibinfo {year} {2014})\BibitemShut {NoStop}%
\bibitem [{\citenamefont {Peruzzo}\ \emph {et~al.}(2014)\citenamefont
  {Peruzzo}, \citenamefont {McClean}, \citenamefont {Shadbolt}, \citenamefont
  {Yung}, \citenamefont {Zhou}, \citenamefont {Love}, \citenamefont
  {Aspuru-Guzik},\ and\ \citenamefont {{O'Brien}}}]{peruzzo2014variational}%
  \BibitemOpen
  \bibfield  {author} {\bibinfo {author} {\bibfnamefont {A.}~\bibnamefont
  {Peruzzo}}, \bibinfo {author} {\bibfnamefont {J.}~\bibnamefont {McClean}},
  \bibinfo {author} {\bibfnamefont {P.}~\bibnamefont {Shadbolt}}, \bibinfo
  {author} {\bibfnamefont {M.~H.}\ \bibnamefont {Yung}}, \bibinfo {author}
  {\bibfnamefont {X.~Q.}\ \bibnamefont {Zhou}}, \bibinfo {author}
  {\bibfnamefont {P.~J.}\ \bibnamefont {Love}}, \bibinfo {author}
  {\bibfnamefont {A.}~\bibnamefont {Aspuru-Guzik}},\ and\ \bibinfo {author}
  {\bibfnamefont {J.~L.}\ \bibnamefont {{O'Brien}}},\ }\bibfield  {title}
  {\bibinfo {title} {A variational eigenvalue solver on a photonic quantum
  processor},\ }\href {https://doi.org/10.1038/ncomms5213} {\bibfield
  {journal} {\bibinfo  {journal} {Nature Communications}\ }\textbf {\bibinfo
  {volume} {5}},\ \bibinfo {pages} {4213} (\bibinfo {year} {2014})}\BibitemShut
  {NoStop}%
\bibitem [{\citenamefont {Bharti}\ \emph {et~al.}(2022)\citenamefont {Bharti},
  \citenamefont {Kyaw}, \citenamefont {Haug}, \citenamefont {Alperin-Lear},
  \citenamefont {Anand}, \citenamefont {Degroote}, \citenamefont {Heimonen},
  \citenamefont {Kottmann}, \citenamefont {Menke}, \citenamefont {Mok},
  \citenamefont {Sim}, \citenamefont {Kwek},\ and\ \citenamefont
  {Aspuru-Guzik}}]{bharti2022noisy}%
  \BibitemOpen
  \bibfield  {author} {\bibinfo {author} {\bibfnamefont {K.}~\bibnamefont
  {Bharti}}, \bibinfo {author} {\bibfnamefont {A.~C.-L. T.~H.}\ \bibnamefont
  {Kyaw}}, \bibinfo {author} {\bibfnamefont {T.}~\bibnamefont {Haug}}, \bibinfo
  {author} {\bibfnamefont {S.}~\bibnamefont {Alperin-Lear}}, \bibinfo {author}
  {\bibfnamefont {A.}~\bibnamefont {Anand}}, \bibinfo {author} {\bibfnamefont
  {M.}~\bibnamefont {Degroote}}, \bibinfo {author} {\bibfnamefont
  {H.}~\bibnamefont {Heimonen}}, \bibinfo {author} {\bibfnamefont {J.~S.}\
  \bibnamefont {Kottmann}}, \bibinfo {author} {\bibfnamefont {T.}~\bibnamefont
  {Menke}}, \bibinfo {author} {\bibfnamefont {W.-K.}\ \bibnamefont {Mok}},
  \bibinfo {author} {\bibfnamefont {S.}~\bibnamefont {Sim}}, \bibinfo {author}
  {\bibfnamefont {L.-C.}\ \bibnamefont {Kwek}},\ and\ \bibinfo {author}
  {\bibfnamefont {A.}~\bibnamefont {Aspuru-Guzik}},\ }\bibfield  {title}
  {\bibinfo {title} {Noisy intermediate-scale quantum algorithms},\ }\href
  {https://doi.org/10.1103/RevModPhys.94.015004} {\bibfield  {journal}
  {\bibinfo  {journal} {Reviews of Modern Physics}\ }\textbf {\bibinfo {volume}
  {94}},\ \bibinfo {pages} {015004} (\bibinfo {year} {2022})}\BibitemShut
  {NoStop}%
\bibitem [{\citenamefont {Park}\ \emph {et~al.}(2020)\citenamefont {Park},
  \citenamefont {Quanz}, \citenamefont {Wood}, \citenamefont {Higgins},\ and\
  \citenamefont {Harishankar}}]{park2020practical}%
  \BibitemOpen
  \bibfield  {author} {\bibinfo {author} {\bibfnamefont {J.-E.}\ \bibnamefont
  {Park}}, \bibinfo {author} {\bibfnamefont {B.}~\bibnamefont {Quanz}},
  \bibinfo {author} {\bibfnamefont {S.}~\bibnamefont {Wood}}, \bibinfo {author}
  {\bibfnamefont {H.}~\bibnamefont {Higgins}},\ and\ \bibinfo {author}
  {\bibfnamefont {R.}~\bibnamefont {Harishankar}},\ }\bibfield  {title}
  {\bibinfo {title} {Practical application improvement to quantum {SVM}: theory
  to practice},\ }\bibfield  {journal} {\bibinfo  {journal} {arXiv:2012.07725}\
  }\href {https://doi.org/10.48550/arXiv.2012.07725}
  {10.48550/arXiv.2012.07725} (\bibinfo {year} {2020})\BibitemShut {NoStop}%
\bibitem [{\citenamefont {Sung}\ \emph {et~al.}(2020)\citenamefont {Sung},
  \citenamefont {Yao}, \citenamefont {Harrigan}, \citenamefont {Rubin},
  \citenamefont {Jiang}, \citenamefont {Lin}, \citenamefont {Babbush},\ and\
  \citenamefont {{McClean}}}]{sung2020using}%
  \BibitemOpen
  \bibfield  {author} {\bibinfo {author} {\bibfnamefont {K.~J.}\ \bibnamefont
  {Sung}}, \bibinfo {author} {\bibfnamefont {J.}~\bibnamefont {Yao}}, \bibinfo
  {author} {\bibfnamefont {M.~P.}\ \bibnamefont {Harrigan}}, \bibinfo {author}
  {\bibfnamefont {N.~C.}\ \bibnamefont {Rubin}}, \bibinfo {author}
  {\bibfnamefont {Z.}~\bibnamefont {Jiang}}, \bibinfo {author} {\bibfnamefont
  {L.}~\bibnamefont {Lin}}, \bibinfo {author} {\bibfnamefont {R.}~\bibnamefont
  {Babbush}},\ and\ \bibinfo {author} {\bibfnamefont {J.~R.}\ \bibnamefont
  {{McClean}}},\ }\bibfield  {title} {\bibinfo {title} {Using models to improve
  optimizers for variational quantum algorithms},\ }\href
  {https://doi.org/10.1088/2058-9565/abb6d9} {\bibfield  {journal} {\bibinfo
  {journal} {Quantum Science and Technology}\ }\textbf {\bibinfo {volume}
  {5}},\ \bibinfo {pages} {044008} (\bibinfo {year} {2020})}\BibitemShut
  {NoStop}%
\bibitem [{\citenamefont {Ito}\ \emph {et~al.}(2021)\citenamefont {Ito},
  \citenamefont {Mizukami},\ and\ \citenamefont {Fujii}}]{ito2021universal}%
  \BibitemOpen
  \bibfield  {author} {\bibinfo {author} {\bibfnamefont {K.}~\bibnamefont
  {Ito}}, \bibinfo {author} {\bibfnamefont {W.}~\bibnamefont {Mizukami}},\ and\
  \bibinfo {author} {\bibfnamefont {K.}~\bibnamefont {Fujii}},\ }\bibfield
  {title} {\bibinfo {title} {Universal noise-precision relations in variational
  quantum algorithms},\ }\bibfield  {journal} {\bibinfo  {journal}
  {{arXiv:2106.03390}}\ }\href {https://doi.org/10.48550/arXiv.2106.03390}
  {10.48550/arXiv.2106.03390} (\bibinfo {year} {2021})\BibitemShut {NoStop}%
\bibitem [{\citenamefont {Rabinovich}\ \emph {et~al.}(2023)\citenamefont
  {Rabinovich}, \citenamefont {Campos}, \citenamefont {Adhikary}, \citenamefont
  {Pankovets}, \citenamefont {Vinichenko},\ and\ \citenamefont
  {Biamonte}}]{rabinovich2023gate}%
  \BibitemOpen
  \bibfield  {author} {\bibinfo {author} {\bibfnamefont {D.}~\bibnamefont
  {Rabinovich}}, \bibinfo {author} {\bibfnamefont {E.}~\bibnamefont {Campos}},
  \bibinfo {author} {\bibfnamefont {S.}~\bibnamefont {Adhikary}}, \bibinfo
  {author} {\bibfnamefont {E.}~\bibnamefont {Pankovets}}, \bibinfo {author}
  {\bibfnamefont {D.}~\bibnamefont {Vinichenko}},\ and\ \bibinfo {author}
  {\bibfnamefont {J.}~\bibnamefont {Biamonte}},\ }\bibfield  {title} {\bibinfo
  {title} {On the gate-error robustness of variational quantum algorithms},\
  }\bibfield  {journal} {\bibinfo  {journal} {{arXiv:2301.00048}}\ }\href
  {https://doi.org/10.48550/arXiv.2301.00048} {10.48550/arXiv.2301.00048}
  (\bibinfo {year} {2023})\BibitemShut {NoStop}%
\bibitem [{\citenamefont {Schuld}\ \emph {et~al.}(2021)\citenamefont {Schuld},
  \citenamefont {Sweke},\ and\ \citenamefont {Meyer}}]{schuld2021effect}%
  \BibitemOpen
  \bibfield  {author} {\bibinfo {author} {\bibfnamefont {M.}~\bibnamefont
  {Schuld}}, \bibinfo {author} {\bibfnamefont {R.}~\bibnamefont {Sweke}},\ and\
  \bibinfo {author} {\bibfnamefont {J.~J.}\ \bibnamefont {Meyer}},\ }\bibfield
  {title} {\bibinfo {title} {Effect of data encoding on the expressive power of
  variational quantum-machine-learning models},\ }\href
  {https://doi.org/10.1103/PhysRevA.103.032430} {\bibfield  {journal} {\bibinfo
   {journal} {Physical Review A}\ }\textbf {\bibinfo {volume} {103}},\ \bibinfo
  {pages} {032430} (\bibinfo {year} {2021})}\BibitemShut {NoStop}%
\bibitem [{\citenamefont {Schuld}\ \emph {et~al.}(2019)\citenamefont {Schuld},
  \citenamefont {Bergholm}, \citenamefont {Gogolin}, \citenamefont {Izaac},\
  and\ \citenamefont {Killoran}}]{Schuld2019ParameterShiftRule}%
  \BibitemOpen
  \bibfield  {author} {\bibinfo {author} {\bibfnamefont {M.}~\bibnamefont
  {Schuld}}, \bibinfo {author} {\bibfnamefont {V.}~\bibnamefont {Bergholm}},
  \bibinfo {author} {\bibfnamefont {C.}~\bibnamefont {Gogolin}}, \bibinfo
  {author} {\bibfnamefont {J.}~\bibnamefont {Izaac}},\ and\ \bibinfo {author}
  {\bibfnamefont {N.}~\bibnamefont {Killoran}},\ }\bibfield  {title} {\bibinfo
  {title} {Evaluating analytic gradients on quantum hardware},\ }\href
  {https://doi.org/10.1103/PhysRevA.99.032331} {\bibfield  {journal} {\bibinfo
  {journal} {Phys. Rev. A}\ }\textbf {\bibinfo {volume} {99}},\ \bibinfo
  {pages} {032331} (\bibinfo {year} {2019})}\BibitemShut {NoStop}%
\bibitem [{\citenamefont {Kingma}\ and\ \citenamefont
  {Ba}(2014)}]{kingdma2014adamOptimizer}%
  \BibitemOpen
  \bibfield  {author} {\bibinfo {author} {\bibfnamefont {D.~P.}\ \bibnamefont
  {Kingma}}\ and\ \bibinfo {author} {\bibfnamefont {J.}~\bibnamefont {Ba}},\
  }\href {https://doi.org/10.48550/ARXIV.1412.6980} {\bibinfo {title} {Adam: A
  method for stochastic optimization}} (\bibinfo {year} {2014})\BibitemShut
  {NoStop}%
\bibitem [{\citenamefont {Nesterov}(2004)}]{nesterov2004introductory}%
  \BibitemOpen
  \bibfield  {author} {\bibinfo {author} {\bibfnamefont {Y.}~\bibnamefont
  {Nesterov}},\ }\href {https://doi.org/10.1007/978-1-4419-8853-9} {\emph
  {\bibinfo {title} {Introductory lectures on convex optimization}}},\ \bibinfo
  {series} {Applied Optimization}, Vol.~\bibinfo {volume} {87}\ (\bibinfo
  {publisher} {Springer Science \& Business},\ \bibinfo {year}
  {2004})\BibitemShut {NoStop}%
\bibitem [{\citenamefont {Abraham}\ and\ \citenamefont
  {et~al.}(2021)}]{qiskit_software}%
  \BibitemOpen
  \bibfield  {author} {\bibinfo {author} {\bibfnamefont {H.}~\bibnamefont
  {Abraham}}\ and\ \bibinfo {author} {\bibnamefont {et~al.}},\ }\href
  {https://doi.org/10.5281/zenodo.2573505} {\bibinfo {title} {Qiskit: An
  open-source framework for quantum computing}} (\bibinfo {year}
  {2021})\BibitemShut {NoStop}%
\bibitem [{\citenamefont {Schmied}(2016)}]{schmied2016qst}%
  \BibitemOpen
  \bibfield  {author} {\bibinfo {author} {\bibfnamefont {R.}~\bibnamefont
  {Schmied}},\ }\bibfield  {title} {\bibinfo {title} {Quantum state tomography
  of a single qubit: comparison of methods},\ }\href
  {https://doi.org/10.1080/09500340.2016.1142018} {\bibfield  {journal}
  {\bibinfo  {journal} {Journal of Modern Optics}\ }\textbf {\bibinfo {volume}
  {63}},\ \bibinfo {pages} {1744} (\bibinfo {year} {2016})}\BibitemShut
  {NoStop}%
\end{thebibliography}%

\end{document}